\documentclass[journal,onecolumn,12pt,twoside]{IEEEtran}

\usepackage{setspace}
\usepackage{mathtools}
\usepackage{subfig}
\usepackage{paralist}
\usepackage{cite}

\usepackage{color}

\usepackage{geometry}
\ifCLASSINFOpdf
\else
\fi

\usepackage[font={small}]{caption}

\usepackage{amsfonts}

\usepackage{amsthm,amssymb}

\newtheorem{corr}{Corollary}
\newtheorem{lem}{Lemma}
\newtheorem{theorem}{Theorem}

\usepackage{mathrsfs}

\usepackage{amsmath}
\usepackage{algorithm}
\usepackage[noend]{algpseudocode}
\makeatletter
\def\BState{\State\hskip-\ALG@thistlm}
\makeatother

\allowdisplaybreaks

\begin{document}
\title{Improving the Coverage and Spectral Efficiency \\of Millimeter-Wave Cellular Networks\\ using Device-to-Device Relays}

\author{\IEEEauthorblockN{Shuanshuan Wu,~\IEEEmembership{Student Member,~IEEE}, Rachad Atat,~\IEEEmembership{Student Member,~IEEE}, \\
Nicholas Mastronarde,~\IEEEmembership{Senior Member,~IEEE}, and
Lingjia Liu,~\IEEEmembership{Senior Member,~IEEE}}

\thanks{S. Wu and N. Mastronarde are with the Dept. of Electrical Engineering, University at Buffalo, The State University of New York, Buffalo, NY 14260 USA (e-mail: \{shuanshu, nmastron\}@buffalo.edu).

R. Atat was with the Electrical Engineering and Computer Science Dept., University of Kansas, Lawrence, KS 66044 USA (e-mail: rachad@ittc.ku.edu).

L. Liu is with the Electrical and Computer Engineering Dept., Virginia Polytechnic Institute and State University, Blacksburg, VA 24061 USA (e-mail: ljliu@vt.edu).
}
}


\maketitle
\vspace{-1cm}
\begin{abstract}
\vspace{-10pt}
The susceptibility of millimeter waveform propagation to blockages limits the coverage of millimeter-wave (mmWave) signals. To overcome blockages, we propose to leverage two-hop device-to-device (D2D) relaying. Using stochastic geometry, we derive expressions for the downlink coverage probability of relay-assisted mmWave cellular networks when the D2D links are implemented in either uplink mmWave or uplink microwave bands. We further investigate the spectral efficiency (SE) improvement in the cellular downlink, and the effect of D2D transmissions on the cellular uplink. For mmWave links, we derive the coverage probability using dominant interferer analysis while accounting for both blockages and beamforming gains. For microwave D2D links, we derive the coverage probability considering both line-of-sight (LOS) and non-line-of-sight (NLOS) propagation. Numerical results show that downlink coverage and SE can be improved using two-hop D2D relaying. Specifically, microwave D2D relays achieve better coverage because D2D connections can be established under NLOS conditions. However, mmWave D2D relays achieve better coverage when the density of interferers is large because blockages eliminate interference from NLOS interferers. The SE on the downlink depends on the relay mode selection strategy, and mmWave D2D relays use a significantly smaller fraction of uplink resources than microwave D2D relays.
\end{abstract}

\vspace{-10pt}
\begin{IEEEkeywords}
\vspace{-10pt}
Millimeter Wave, Cellular, D2D, relay, coverage, spectral efficiency, stochastic geometry.
\end{IEEEkeywords}

\IEEEpeerreviewmaketitle

\section{Introduction}
\IEEEPARstart{M}{illimeter-wave} (mmWave) communication is attracting considerable attention from the scientific community and regulators for its potential to fulfill the ever increasing demands for mobile broadband access \cite{rappaport:mmwave_will_work}. In industry, the $\text{3}^\textup{rd}$ Generation Partnership Project (3GPP) has started a study item focusing on the performance and feasibility of wireless communication in the mmWave band (6-100 GHz) \cite{3gpp:mmwave_TR}. Moreover, in July 2016, the Federal Communications Commission (FCC) unanimously voted to open nearly 11 GHz of mmWave spectrum for 5G~\cite{fcc:mmwave}. However, due to its weak diffraction ability and severe penetration loss, mmWave communication is highly susceptible to blockages. For example, measurements on mmWave propagation show that tinted glass and brick pillars have high penetration losses of $40.1$ dB and $28.3$ dB, respectively, at a frequency of $28$ GHz \cite{zhao:28ghz}. While in indoor environments, mmWave links experience intermittent connectivity due to blockages from mobile human bodies~\cite{collonge:60ghz}. In other words, transmissions in mmWave networks require line-of-sight (LOS) paths, causing pronounced coverage holes.

A straightforward solution to overcome blockages in mmWave cellular networks is to deploy more mmWave base stations (BSs). Although mmWave systems are expected to leverage highly directional steerable beam antenna arrays to extend their transmission range and to reduce intercell interference from off-boresight directions~\cite{rappaport:mmwave_propagation}, significant interference may still be experienced in dense BS deployments. For example, the signal-to-interference-plus-noise ratio (SINR) at a user equipment (UE) that is located in the LOS path of a directional beam from an interfering BS will be severely degraded~\cite{niu:mmwave_survey}. The likelihood of this happening increases with the BS density, so interference management/coordination becomes a challenge in dense mmWave networks. The analytical and simulation results in this paper verify this intuition.

Another promising approach to extend coverage in mmWave cellular networks is to allow an intermediate relay node to forward traffic from a BS to a destination UE, which has poor links to nearby BSs. Due to the potential to route around blockages, using two-hop relay transmissions can improve mmWave coverage. In general, the relay transmission could be completed in the mmWave or microwave spectrum, and a relay node can either be deployed by an operator, e.g., a so-called \textit{infrastructure relay} in long term evolution (LTE) \cite{3gpp:relay}, or can be an idle UE that is used opportunistically. This latter case is attractive because it does not drastically change the network topology or infrastructure requirements. Since a UE serving as a relay connects to a destination UE via a device-to-device (D2D) link, we refer to it as a \textit{D2D relay}.

\subsection{Related Work}
D2D communications is already playing an important role in the unlicensed band via Wi-Fi Direct~\cite{wi-fi-direct}. However, its counterpart in the licensed band is far from being fully developed. A preliminary version of D2D communications called Proximity Service (ProSe)~\cite{lte-direct} is standardized in LTE-Advanced, but is mainly targeted for public safety use~\cite{feng:device}. Aside from public safety use cases, D2D communications can improve spectral efficiency in microwave cellular networks up to $4$ - $5\times$~\cite{wu:flashlinq}. Moreover, it opens up opportunities for social networking~\cite{corson:prose}, multicasting~\cite{zhou:intracluster}, machine type communications~\cite{pratas:underlay}, and D2D content distribution~\cite{golrezaei:cache}. There is also work focusing on D2D communications in mmWave spectrum~\cite{qiao:mmwave_d2d}.

D2D relaying to extend network coverage for public safety use is introduced in Release-13 LTE-Advanced~\cite{3gpp:prose,3gpp:36300}. In~\cite{mastronarde:to_relay}, we address the problem of incentivizing UEs to relay in commercial cellular networks. Currently, there is only limited research on D2D relay-assisted mmWave cellular communications~\cite{wu:mmwave,lin:connectivity,wei:mmwave_d2d_relay,turgut:relay_energy,xie:relay_assisted}.
In \cite{lin:connectivity}, the authors use stochastic geometry to analyze the connectivity of mmWave networks with multi-hop relaying. In particular, they derive the probability that there exists at least one path to route a directional mmWave signal from a source to a destination with the potential help of relays. They demonstrate that multi-hop relaying can improve connectivity, but do not analyze the coverage probability or spectral efficiency. Note that, unlike connectivity, coverage depends on the SINR. 
In~\cite{wei:mmwave_d2d_relay}, the authors investigate the probing and relay selection problem in two-hop relay-assisted mmWave cellular networks using an i.i.d. Bernoulli obstacle model and, in~\cite{turgut:relay_energy}, the authors analyze the energy efficiency of relay-assisted mmWave cellular networks. However, neither~\cite{wei:mmwave_d2d_relay} nor~\cite{turgut:relay_energy} analyze the coverage probability or spectral efficiency.
In~\cite{xie:relay_assisted}, the authors analyze the coverage probability of relay-assisted mmWave cellular networks assuming that the UE is associated with the nearest BS or, if the nearest BS is non-line-of-sight (NLOS), then it associates with the nearest relay. In contrast, we consider a two-hop relaying model in which mode selection depends on the BS-UE, BS-relay and relay-UE link signal qualities. 

In our recent work~\cite{wu:mmwave}, we took steps towards filling the gaps in prior research by analyzing the coverage probability of D2D relay-assisted mmWave cellular networks, where the D2D link operates in mmWave spectrum. In this paper, we extend~\cite{wu:mmwave} in the following directions: we analyze the network's coverage when the D2D link uses microwave spectrum; we analyze the coverage of noise-limited mmWave links; we derive the spectral efficiency (SE) of the downlink and analyze the effect of D2D transmission on the uplink resources; and we more fully explore the derived models through analytical and simulation results. 

In 3GPP evaluation scenarios, BSs are assumed to be arranged in hexagonal or square grids and system performance is evaluated using simulations. However, grid models are known to be too ideal because actual BS deployments are restricted by terrain and obstacles. Another shortcoming of grid BS models is their lack of tractability. Recently, stochastic geometry has gained popularity as a powerful tool for modeling/evaluating wireless networks \cite{chiu:stochastic_book,baccelli:stochastic_wireless_book} not only because it leads to a mathematically tractable framework, but also because positions of BSs in real networks often resemble point processes~\cite{lee:bs_position,andrews:tractable}, which lie at the heart of stochastic geometry. More specifically, in stochastic geometry analysis, spatial locations of network objects, e.g., BSs, UEs, and obstacles, are modeled as Poisson point processes (PPPs) and performance metrics are derived by averaging over their potential topological realizations~\cite{chiu:stochastic_book}. This provides a tractable approach for analyzing network performance \cite{andrews:tractable}.

\subsection{Contributions}
In this paper, we analyze the downlink performance of a two-hop relay-assisted mmWave cellular network using stochastic geometry. 
In the considered network, a downlink transmission is switched from direct cellular mode to D2D relay mode if there is an outage of the cellular link, but a relay UE is available that can help complete the transmission from the BS to the destination UE. For D2D transmissions (from relay UEs to destination UEs), both mmWave and microwave D2D are possible options. Intuitively, mmWave D2D is likely to achieve higher data rates due to the increased bandwidth available in the mmWave band, while microwave D2D is expected to achieve better coverage in dense blockage scenarios because of the better propagation properties in the microwave band. Therefore, we study and compare both D2D relay schemes. Here, the microwave band refers to the sub-6 GHz spectrum used in conventional cellular networks.

Our contributions are as follows:
\begin{itemize}
\item We derive the LOS probability between any two network nodes assuming cylindrical obstacles distributed according to a 2D homogeneous Poisson point process (PPP).
\item We show that the downlink coverage probability of a relay-assisted mmWave cellular network depends on the coverage of the direct cellular and D2D links, which are independent when D2D is deployed on the uplink cellular spectrum.
\item We derive the coverage probabilities of mmWave cellular and mmWave D2D links using dominant interferer analysis considering both blockages (based on the derived LOS probability model) and beamforming gains (obtained using square antenna arrays). 
\item We investigate the coverage of noise-limited mmWave links. The noise-limited assumption significantly simplifies the coverage probability model and, according to our numerical results, is accurate for lower BS densities and for higher obstacle densities.
\item We derive the coverage probability for microwave D2D links considering both multi-path fading and path loss. 
We use different path loss models for different LOS and NLOS propagation scenarios and derive the coverage probability using the law of total probability.
\item We derive the spectral efficiency of the relay-assisted mmWave cellular downlink. We also study the effect of D2D transmissions on the cellular uplink resources.
\item We validate our analytical results against simulations based on 3GPP network evaluation scenarios and channel models. We then explore the effect of different parameters on the performance of relay-assisted mmWave cellular networks. Our results demonstrate that two-hop D2D relays can improve coverage and SE across a variety of network configurations and that microwave D2D relays achieve higher performance gains than mmWave D2D relays, except under extremely dense BS deployments, i.e., dense interferers.
\end{itemize}

The remainder of this paper is organized as follows. In Section~\ref{sec:system_model}, we present the system model. In Section~\ref{sec:coverage_analysis}, we derive the coverage probability for two-hop relay-assisted mmWave communications using stochastic geometry. In Section~\ref{sec:spectral}, we analyze the SE improvement in the cellular downlink and the effect of D2D transmissions on the cellular uplink resources. In Section~\ref{sec:simulation}, we present our numerical results. 
We conclude in Section~\ref{sec:conclusion}.

\section{System Model}
\label{sec:system_model}

\subsection{Geometric Assumptions}
\subsubsection{Poisson point processes}

In this paper, we model the spatial locations of BSs, UEs, and obstacles as 2D homogeneous PPPs. A PPP defined in $\mathbb{R}^2$ is a random process in which the number of points $\Phi$ in a bounded Borel set $\mathcal{B} \subset \mathbb{R}^2$ has a Poisson distribution:
\begin{equation}
\mathbb{P}(\Phi(\mathcal{B})=k) = \frac{\Lambda^k}{k!} e^{-\Lambda}, k = 0, 1, 2, ...
\end{equation}
where $\Lambda = \lambda \mathsf{v}_2(\mathcal{B})$ is the expectation of the Poisson random variable for some intensity $\lambda$ and $\mathsf{v}_2(\cdot)$ denotes the Lebesgue measure in $\mathbb{R}^2$. If $\lambda$ is constant, the PPP is said to be homogeneous.

\subsubsection{Obstacles}
In general, obstacles have arbitrary shapes. In outdoor environments, major obstacles are buildings that can be approximated by rectangles or polygons. In indoor environments, human bodies are common obstacles that can be modeled as cylinders. In \cite{bai:coverage}, blockage effects of rectangular obstacles are investigated. In our analysis, to derive a single tractable LOS probability model for both outdoor and indoor scenarios, we model obstacles as cylinders distributed according to a 2D homogeneous PPP $\Phi_\text{o}$ with intensity $\lambda_\text{o}$. In Section~\ref{sec:los_validate}, we validate our decision to use this model for outdoor scenarios. 

\subsubsection{BSs and UEs}
We model the distribution of BSs as a 2D homogeneous PPP $\Phi_\text{b}$ with intensity $\lambda_\text{b}$. We assume that UEs are also distributed according to a 2D homogeneous PPP and that they are partitioned between active and idle UEs, where idle UEs are candidate relays. The candidate relay UEs can be determined by independent thinning on the set of all UEs. We use $\Phi_\text{r}$ and $\lambda_\text{r}$ to denote the point process and intensity of candidate relay UEs, respectively. Since the BSs, candidate relay UEs, and obstacles form homogeneous PPPs, we focus our analysis on a typical destination UE (in either the cellular or D2D link), which we assume is located at the origin $\boldsymbol{o}$. This is permissible in a homogeneous PPP by Slivnyak's theorem \cite{chiu:stochastic_book}. 

\subsection{MmWave Beamforming}
\label{sec:antenna_array}
As noted in the introduction, mmWave systems are expected to leverage highly directional beams to extend their transmission range \cite{rappaport:mmwave_will_work,roh:mmwave_beamforming}. In this paper, we consider a simple sectored antenna array model for both mmWave transmitters and receivers \cite{bai:coverage_2015}. In the mmWave links, $\phi_\text{b}$ and $\phi_\text{u}$ denote the half-power beamwidths of BSs and UEs, respectively. The main lobe gain and side lobe gain of the BSs are denoted by $Gm_\text{b}$ and $Gs_\text{b}$, respectively. For receiving UEs, the main and side lobe gains are denoted by $Gm_\text{u}$ and $Gs_\text{u}$, respectively. With the assumption of an $N \times N$ uniform planar square antenna array with half-wavelength antenna spacing, the beamwidth $\phi$, main lobe gain $G_m$, and side lobe gain $G_s$ are given as \cite{venugopal:d2d_mmwave}
\begin{equation}
\label{eq:antenna_parameter}
\begin{split}
\phi = 1.732/N, ~G_m = N^2, ~G_s = 1/\sin^2\left(\frac{3\pi}{2N}\right).
\end{split}
\end{equation}
The antenna array gain $g_i$ from an arbitrary BS $i$ to the typical cellular receiver is shown in~\eqref{eq:interference_gain_cell}. For simplicity, we denote the four possible antenna gains in \eqref{eq:interference_gain_cell} as $G_k$, $k=1,2,3,4$. Note that the beamforming gain on the desired cellular link is always the \textit{on-boresight gain} $G_1 = Gm_\text{b} Gm_\text{u}$. The antenna array gains on a D2D link can be acquired similarly. For tractability, we assume that an interfering transmitter's antenna boresight is uniformly distributed over $\left[0, 2\pi\right)$.
\begin{equation}
\label{eq:interference_gain_cell}
g_i = 
\begin{cases}
Gm_\text{b} Gm_\text{u}, &\text{with probability } \frac{\phi_\text{b}}{2\pi} \frac{\phi_\text{u}}{2\pi}, \\
Gm_\text{b} Gs_\text{u}, &\text{with probability } \frac{\phi_\text{b}}{2\pi} (1 -\frac{\phi_\text{u}}{2\pi}), \\
Gs_\text{b} Gm_\text{u}, &\text{with probability } (1-\frac{\phi_\text{b}}{2\pi}) \frac{\phi_\text{u}}{2\pi}, \\
Gs_\text{b} Gs_\text{u}, &\text{with probability } (1-\frac{\phi_\text{b}}{2\pi}) (1-\frac{\phi_\text{u}}{2\pi}).
\end{cases}
\end{equation}

\subsection{LOS Probability}
\label{sec:los_prob}
As mentioned earlier, we model obstacles as cylinders that are spatially distributed according to a 2D homogeneous PPP $\Phi_\text{o}$ with intensity $\lambda_\text{o}$. Let the random variable $H \in [h_\text{min}, h_\text{max}]$ denote the obstacle height with probability density function (PDF) $f_H(h)$, and let $h_\text{Tx}$ and $h_\text{Rx}$ denote the transmitter and receiver antenna heights, respectively. We have the following conclusion regarding the LOS probability between the two nodes.
\begin{lem}[LOS Probability]
\label{lem:los_prob}
\textup{
The LOS probability between two nodes separated by a distance $d$ on the  plane is $p_{\text{L}}(d) = \exp \big(-\eta \lambda_\text{o}(2 \mathbb{E}[R] d + \pi \mathbb{E}[R^2]) \big)$, where $\lambda_\text{o}$ is the obstacle intensity, the random variable $R$ denotes the obstacle radius, and $\eta = 1 - \int_0^1 \int_{h_\text{min}}^{s h_\text{Rx} + (1-s)h_\text{Tx}} f_H(h)\mathrm{d}h\mathrm{d}s$ is the obstacle thinning factor, which accounts for obstacle heights.
}
\end{lem}
\begin{proof}
The proof can be found in~\cite{wu:mmwave}.
\end{proof}

In the rest of this paper, we will express the LOS probability as:
\begin{equation}
\label{eq:los_prob_simplify}
p_{\text{L}}(d) = c e^{-\beta d},
\end{equation}
where $c = e^{-\eta\lambda_\text{o} \pi \mathbb{E}[R^2]}$ and $\beta = 2 \eta\lambda_\text{o} \mathbb{E}[R]$. Note that $\eta$ may be different for D2D and cellular links because they typically have different transmitter antenna heights. We denote these different parameters as $\eta_\text{D}$ and $\eta_\text{C}$, respectively.

\subsection{Channel Models}
\label{sec:channel_models}
In this paper, we consider the following path loss model for all link types (cellular and D2D) and for all spectrum bands (mmWave and microwave):
\vspace{-8pt}
\begin{equation*}
PL = A_1 \log_{10}(d) + A_2 + A_3 \log_{10}(f_c) + X \quad \text{(dB)},
\vspace{-8pt}
\end{equation*}
where $d$ (m) is the distance between the transmitter and receiver, $f_c$ (GHz) is the carrier frequency, $A_1$ includes the path loss exponent, $A_2$ is the intercept, $A_3$ describes the path loss frequency dependence, and $X$ is an environment-specific term, for example, it can be used to describe the wall attenuation. Path loss in linear scale can be expressed as $PL = A \cdot d^{\alpha}$, where $A = 10^{(A_2+X)/10} f_c^{A_3/10}$ and $\alpha = A_1/10$ is the path loss exponent. The values of $A_1$, $A_2$, $A_3$ and $X$ depend on the deployment scenario. In the remainder of this paper, we will use the notation $A_\text{L}$ and $A_\text{N}$ to differentiate between LOS and NLOS path losses on microwave D2D links.

We now describe the considered mmWave and microwave channel models in detail.

\subsubsection{MmWave Cellular and D2D Links}
Due to the properties of mmWave propagation, multi-path effects are negligible. For example, at 60 GHz, the channel closely matches an Additive White Gaussian Noise (AWGN) channel \cite{niu:mmwave_survey}. Consequently, we do not consider multi-path fading in our mmWave channel model. 
Interference experienced by a mmWave link comes from transmitters that have LOS paths to the receiver. On the downlink, interferers are BSs of neighboring cells. On the D2D link, interferers are cellular uplink users and D2D transmitters from other cells that happen to have been scheduled on the same frequency. Thus, the received signal power over a mmWave link of length $d$ is
\vspace{-8pt}
\begin{equation}
\label{eq:signal_model_mm}
P_\text{Rx} = B P_\text{Tx} G_\text{Tx} G_\text{Rx} A^{-1} d^{-\alpha},
\vspace{-8pt}
\end{equation}
where $G_\text{Tx}$ and $G_\text{Rx}$ denote the antenna array gains at the transmitter and receiver, respectively; $P_\text{Tx}$ denotes the transmit power; and $B$ is a Bernoulli random variable with parameter $p_\text{L}(d)$.

\subsubsection{Microwave D2D Links}
Multi-path effects play an important role in microwave signal propagation. For tractability, we assume microwave D2D signals experience Rayleigh fading such that the instantaneous channel power gain $h$ is exponentially distributed with PDF $f_h(x) = \mu e^{-\mu x}$ and mean $1/\mu$. In addition to Rayleigh fading, we consider both LOS and NLOS propagation by using different path loss models for each case. Note that in microwave communications, Rayleigh fading is used when there is no dominant LOS path and Rician fading is typically used when there is a dominant LOS path. However, since large scale fading contributes significantly more to the signal power attenuation, we incorporate the LOS/NLOS effect into the path loss model instead of complicating our analysis with Rician fading.\footnote{We have validated this approximation in several scenarios under different parameter configurations against simulations using Rician fading. We omit these results due to space limitations.} The received power over a microwave D2D link with length $d$ is modeled as $P_\text{Rx} = P_\text{Tx} h A^{-1} d^{-\alpha}$. 

\subsection{UE Association}
\label{sec:ue_association} 
To be consistent across the different types of links considered in this paper (i.e., mmWave cellular, mmWave D2D, and microwave D2D), we assume that the receiver UE always associates with the transmitter that has the smallest path loss. Here, ``transmitter'' refers to the mmWave BS on a cellular link or the relay UE on a D2D link.

Since we assume that mmWave transmission is only possible when a LOS path exists, the above criterion indicates that the UE always associates with the nearest LOS BS in cellular mode or the nearest LOS candidate relay UE in mmWave D2D relay mode. Let $d_{0,\text{L}}$ denote the distance between a typical UE and the nearest LOS transmitter and let $\lambda$ denote the intensity of candidate transmitters, i.e., $\lambda=\lambda_\text{b}$ for cellular links and $\lambda=\lambda_\text{r}$ for D2D links. The PDF of $d_{0,\text{L}}$ is given in the following lemma.
\begin{lem}[PDF of the distance $d_{0,\text{L}}$]
\label{lem:pdf_d0_mmwave}
\textup{
The PDF of the random variable $d_{0,\text{L}}$ is
\begin{equation}
\label{eq:pdf_d0_los}
f_{d_{0,\text{L}}}(d) = 2\pi\lambda d c e^{-\Lambda(d,\lambda) - \beta d},
\end{equation}
where $\Lambda(d,\lambda)$ is given as
\vspace{-8pt}
\begin{equation}
\label{eq:lambda_big_generic}
\Lambda(d,\lambda) =  \frac{2 \pi \lambda c}{\beta^2} ( 1 - e^{-\beta d} - \beta d e^{-\beta d} ).
\vspace{-8pt}
\end{equation}
}
\end{lem}
\begin{proof}
The proof can be found in~\cite{wu:mmwave}.
\end{proof}

For microwave D2D links, the typical UE associates with the nearest LOS candidate relay UE or the nearest NLOS candidate relay UE, depending on which one provides the smallest path loss. Let $d_{0,\text{N}}$ denote the distance between the typical UE and the nearest NLOS candidate relay UE. The PDF of $d_{0,\text{N}}$ is given in the following lemma.
\begin{lem}[PDF of the distance $d_{0,\text{N}}$]
\textup{
The PDF of the random variable $d_{0,\text{N}}$ is
\vspace{-8pt}
\begin{equation}
\label{eq:pdf_d0_nlos}
f_{d_{0,\text{N}}}(d) = 2\pi\lambda d (1-c e^{-\beta d}) e^{\Lambda(d,\lambda) -\pi\lambda d^2},
\vspace{-8pt}
\end{equation}
where $\Lambda(d,\lambda)$ is defined in \eqref{eq:lambda_big_generic}.
}
\end{lem}
\begin{proof}
The proof is similar to the proof of Lemma~\ref{lem:pdf_d0_mmwave} and is omitted to save space.
\end{proof}

Based on the above discussion, operating in mmWave D2D relay mode requires the relay UE to have LOS paths to both the mmWave BS and the destination UE. On the other hand, to operate in microwave D2D relay mode, the relay UE must have a LOS path to the mmWave BS, but does not require a LOS path to the destination UE. Intuitively, we expect the latter case to achieve better coverage due to the relaxed requirements on the relay UEs. The three transmission modes investigated in this work are shown in Fig.~\ref{fig:tx_mode}. Note that we assume D2D transmissions share the cellular system's uplink spectrum resources, i.e., uplink mmWave spectrum or uplink microwave spectrum.\footnote{Deploying D2D communication in the uplink spectrum is known to provide higher spectral efficiency than in the downlink spectrum because there is relatively lower interference and it is often underutilized \cite{3gpp:36300} \cite{lin:d2d_overview}. Furthermore, we assume that the network employs either microwave D2D relays or mmWave D2D relays, but we do not consider the case that both D2D relay modes are deployed in the same cellular network at the same time. We plan to investigate the latter case in future research.} A list of frequently used notation is provided in Table~\ref{table:notations}.

\begin{figure*}
\centering
\includegraphics[width=0.5\textwidth]{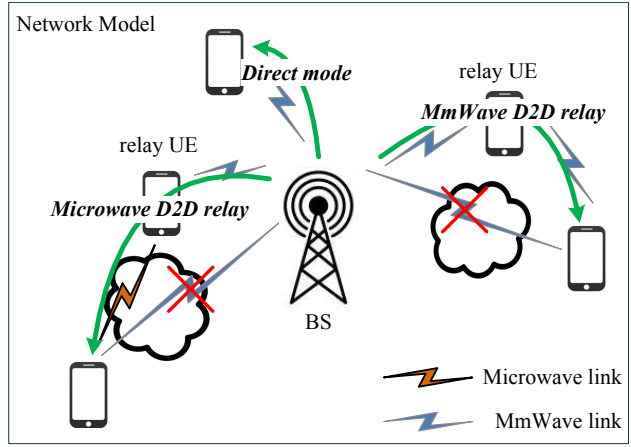}
\caption{\label{fig:tx_mode} Three transmission modes in the D2D relay-assisted cellular network: cellular/direct mode, mmWave D2D relay mode, and microwave D2D relay mode.}
\end{figure*}

\begin{table}
\small
\vspace{0.7in}
\caption{\label{table:notations}List of abbreviated notation}
\newcommand{\tabincell}[2]{\begin{tabular}{@{}#1@{}}#2\end{tabular}}
\centering 
\begin{tabular}{|c|l|}
\hline
\textbf{Notation} & \textbf{Description} \\
\hline
$\Phi_\text{b}$, $\Phi_\text{o}$, $\Phi_\text{r}$ & PPPs of BSs, obstacles, and candidate relay UEs\\
\hline
$\lambda_\text{b}$, $\lambda_\text{o}$, $\lambda_\text{r}$ & Intensities of BS, obstacle, and candidate relay UE PPPs\\
\hline
$p_\text{L}(x)$ & LOS probability of a link with distance $x$\\
\hline
$p_\text{c}(\tau)$, $p_\text{c,C}(\tau)$, $p_\text{c,D}(\tau)$, $p_\text{c,R}(\tau)$ & Overall, cellular, D2D, and relay coverage probabilities at SINR threshold $\tau$\\
\hline
$R$, $H$ & Random variables of obstacle radius and height\\
\hline
$d_{0,\text{L}}$, $f_{d_{0,\text{L}}}(x)$ & Distance from the typical UE to the nearest LOS transmitter and its PDF\\
\hline
$d_{0,\text{N}}$, $f_{d_{0,\text{N}}}(x)$ & Distance from the typical UE to the nearest NLOS transmitter and its PDF\\
\hline
$\boldsymbol{B}(\boldsymbol{o}, d)$ & A disc with radius $d$ and origin $\boldsymbol{o}$\\
\hline
\end{tabular}
\vspace{-0.4in}
\end{table}

\section{Coverage Probability Analysis}
\label{sec:coverage_analysis}

Let $p(\tau) \triangleq \mathbb{P}(\text{SINR} > \tau)$ denote the coverage probability in terms of outage SINR threshold $\tau$. In the considered D2D relay-assisted cellular network, the downlink transmission switches from cellular to D2D relay mode when the direct cellular link has SINR smaller than an outage threshold $\tau$ and there exists an idle UE such that both the cellular link (from the BS to the idle UE) and the D2D link (from the idle UE to the destination UE) have SINRs above $\tau$. On the other hand, an outage occurs if the direct cellular link experiences an outage and such an idle UE does not exist. It follows that the overall downlink coverage probability can be expressed as $p_\text{c}(\tau) = 1- [1-p_\text{c,C}(\tau)][1-p_\text{c,R}(\tau)]$, 
where $p_\text{c,C}(\tau)$ is the coverage probability of the cellular downlink and $p_\text{c,R}(\tau)$ is the relay-assisted two-hop coverage probability. Because the cellular downlink and D2D transmissions are completed in orthogonal frequency bands, i.e., one is on the downlink spectrum and one is on the uplink spectrum, we model the coverage of each link as independent\footnote{Note that in practice, cellular and D2D link coverage may be correlated because the two links may share obstacle(s). For simplicity, we ignore this source of correlation. However, we will consider it in future work.}; therefore, $p_\text{c,R}(\tau) = p_\text{c,C}(\tau) p_\text{c,D}(\tau)$, where $p_\text{c,C}(\tau)$ and $p_\text{c,D}(\tau)$ are the coverage probabilities of the cellular downlink and D2D link, respectively. It follows that the overall downlink coverage probability in a D2D relay-assisted cellular network can be expressed as
\begin{equation}
\label{eq:cov_prob}
\begin{split}
p_\text{c}(\tau) =& 1- [1-p_\text{c,C}(\tau)][1-p_\text{c,C}(\tau) p_\text{c,D}(\tau)] \\
=& p_\text{c,C}(\tau) [1+p_\text{c,D}(\tau)] - p_\text{c,C}^2(\tau) p_\text{c,D}(\tau).
\end{split}
\end{equation}

We derive the coverage probability of the mmWave cellular downlink, $p_\text{c,C}(\tau)$, in Section~\ref{sec:coverage_cellular}. Subsequently, we derive the coverage probability, $p_\text{c,D}(\tau)$, for mmWave D2D links in Section~\ref{sec:coverage_d2d_mmwave} and for microwave D2D links in Section~\ref{sec:coverage_d2d_microwave}. 

\subsection{Coverage Probability of MmWave Cellular Links}
\label{sec:coverage_cellular}

The desired signal at the mmWave receiver will experience interference if there is a LOS path to an interfering BS. However, the interference power depends on the distance and the antenna boresights of both the receiver and the interferer. The SINR at a typical receiver is:
\begin{equation}
\label{eq:sinr_definition}
\text{SINR}_0 = \frac{g_0 A^{-1} d_0^{-\alpha}}{\sigma^2 + \sum_{i>0} B_i g_i A^{-1} d_i^{-\alpha}},
\end{equation}
where $d_0$ is the distance between the typical receiver and the serving BS, which is distributed according to the PDF in \eqref{eq:pdf_d0_los}, $B_i$ is a Bernoulli random variable with parameter $p_\text{L}(d)$, and $d_i$ is the distance between the receiver and BS $i$. The transmit power is incorporated in $\sigma^2$. 

As discussed earlier a typical UE always associates to the nearest LOS BS, which has the smallest path loss. Therefore, all other BSs with LOS paths to the typical UE are interferers that are farther from the typical UE than its associated BS. We partition these interferers into two subsets: \textit{dominant} and \textit{non-dominant} interferers \cite{weber:overview}. A dominant interferer can cause an outage at the receiver, whereas a non-dominant interferer only contributes marginally to the interference.

Let $I = \sum_{i>0} B_i g_i A^{-1} d_i^{-\alpha}$ denote the aggregate interference at the typical receiver. Given the SINR threshold $\tau$ and the distance $d_0$ between the typical UE and its nearest LOS BS, the coverage probability is:
\begin{equation}
\label{eq:cov_mm_d2d_bound}
\begin{split}
p_\text{c}(\tau, d_0) &= \mathbb{P} \left( g_0 A^{-1} d_0^{-\alpha}/\left(\sigma^2+I\right) > \tau \right) \\
&= \mathbb{P} \left(I < \left(g_0 A^{-1} d_0^{-\alpha} - \tau \sigma^2 \right)/\tau \right) \\
& \leq \mathbb{P} \left(d_i^{-\alpha} < \left(g_0 d_0^{-\alpha} - \tau A \sigma^2\right)/\left(g_i \tau\right) \right) \\
&= \mathbb{P} \left(d_i > \left( g_i \tau /\left(g_0 d_0^{-\alpha} - \tau A \sigma^2 \right) \right) ^{1/\alpha} \right).\\
\end{split}
\end{equation}
In \eqref{eq:cov_mm_d2d_bound}, the inequality follows from the fact that the third line only considers the interference from BS $i$ rather than the aggregate interference. We denote by $D_\text{I}(\tau, d_0, g_i) = \left( \frac{g_i \tau}{g_0 d_0^{-\alpha} - \tau A \sigma^2} \right) ^{1/\alpha}$ the distance from the typical UE to the boundary of a region around the typical UE where dominant interferers can exist. We refer to this bounded region as the \textit{interference region (IR)}.

As we have seen in Section~\ref{sec:antenna_array}, the antenna array gain $g_i$ on an interfering link depends on the boresights of antenna arrays at both the interferer and the receiver, which means that the boundary of the IR varies with the direction. To be more specific, $D_\text{I}(\tau, d_0, g_i)$ can take four different values since there are four possible antenna array gains:
\begin{equation}
\label{eq:radius_ir}
\begin{split}
D_k(\tau, d_0) = \left( \frac{ G_k \tau}{g_0 d_0^{-\alpha} - \tau A \sigma^2} \right) ^{1/\alpha}, \ k = 1,2,3,4,
\end{split}
\end{equation}
where $G_k$ is given in \eqref{eq:interference_gain_cell}. An IR is shown in Fig.~\ref{fig:interference_region_cellular}. Note that since interferers are farther than the serving BS, the IR actually excludes a disc $\boldsymbol{B}(\boldsymbol{o},d_0)$ centered at the typical UE. 

\begin{figure}
\centering
\includegraphics[width=0.5\textwidth]{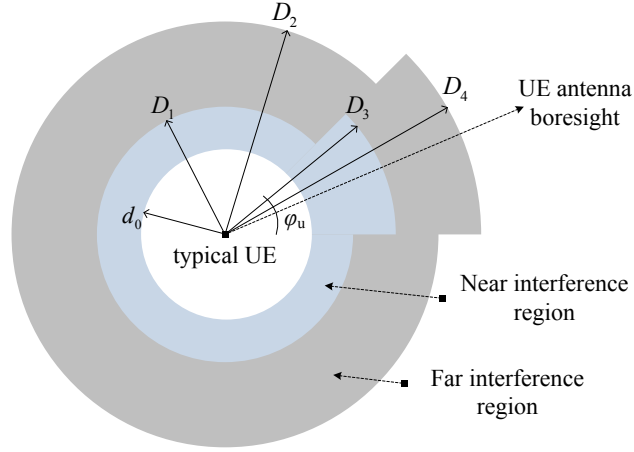}
\caption{\label{fig:interference_region_cellular}Interference region with antenna array gain considered.}
\end{figure} 

According to (\ref{eq:radius_ir}), we further partition the IR into two parts: the \textit{near interference region (NIR)} and the \textit{far interference region (FIR)}. All LOS interferers in the NIR are dominant interferers. However, in the FIR, only LOS interferers with their main lobes towards the typical UE are considered dominant. We thus have the following result.

\begin{theorem}[MmWave cellular coverage probability]
\label{theorem:cov_mmwave_cellular}
\textup{
Given the outage SINR threshold $\tau$, the coverage probability of a mmWave cellular link, $p_\text{c,C}(\tau)$, is given by
\begin{equation}
\label{eq:cov_mmwave_cellular}
p_\text{c,C}(\tau) = 2\pi\lambda_\text{b} c \int_{x>0} x e^{ -\Lambda_\text{b}^{\text{(N)}}(\tau, x) -\Lambda_\text{b}^{\text{(F)}}(\tau, x) -\Lambda(x,\lambda_\text{b}) -\beta x}\mathrm{d}x,
\end{equation}
where $\Lambda(x,\lambda_\text{b})$ is defined in \eqref{eq:lambda_big_generic}, and $\Lambda_\text{b}^{\text{(N)}}(\tau, x)$ and $\Lambda_\text{b}^{\text{(F)}}(\tau, x)$ are given as:
\begin{equation*}
\begin{split}
\Lambda_\text{b}^{\text{(N)}}(\tau, x) =& \frac{ \lambda_\text{b} c}{\beta^2} \Big( \phi_\text{u} \big( \beta x e^{-\beta x} - \beta D_3 e^{-\beta D_3} + e^{-\beta x} - e^{-\beta D_3} \big) \\
& + (2\pi -\phi_\text{u}) \big( \beta x e^{-\beta x} - \beta D_1 e^{-\beta D_1} + e^{-\beta x} - e^{-\beta D_1} \big) \Big),\\
\Lambda_\text{b}^{\text{(F)}}(\tau, x) =& \frac{ \phi_\text{b} \lambda_\text{b} c}{2\pi \beta^2} \Big( \phi_\text{u} \big( \beta D_3 e^{-\beta D_3} - \beta D_4 e^{-\beta D_4} + e^{-\beta D_3} - e^{-\beta D_4} \big) \\
& + (2\pi -\phi_\text{u}) \big( \beta D_1 e^{-\beta D_1} - \beta D_2 e^{-\beta D_2} + e^{-\beta D_1} - e^{-\beta D_2} \big) \Big),\\
\end{split}
\end{equation*}
in which $D_k$, $k=1, 2, 3, 4$, is shorthand for $D_k(\tau, x)$ as defined in \eqref{eq:radius_ir}.
}
\end{theorem}
\begin{proof}
The proof is provided in Appendix~\ref{app:proof_cov_mmwave_cellular}.
\end{proof}

Note that \eqref{eq:cov_mmwave_cellular} provides an upper bound on the coverage probability of PPP-based BS models because dominant interferer analysis uses a lower bound on the interference. We expect that the upper bound is tight for mmWave because the LOS probability exponentially decreases with distance and square antenna arrays reject interference from off-boresight directions,
thereby reducing the effective number of distant interferers. We validate this in Section V.

\subsection{Coverage Probability of MmWave D2D Links}
\label{sec:coverage_d2d_mmwave}

We still use dominant interferer analysis to derive the coverage probability of mmWave D2D links. However, to account for interference, we introduce a multiplexing factor $\rho$ such that uplink inteferers can be approximated by a homogeneous PPP with intensity $\rho\lambda_\text{b}$~\cite{wu:mmwave}.\footnote{Note that uplink interferers do not necessarily form a homogeneous PPP if the BSs form a homogeneous PPP \cite{andrews:mmwave}. However, it has been shown that the dependencies between scheduled UEs in the uplink and their associated BSs are very weak~\cite{novlan:uplink}.} We have the following result. The proof is omitted because it is similar to that of Theorem~\ref{theorem:cov_mmwave_cellular}.
\begin{theorem}[MmWave D2D coverage probability]
\label{theorem:cov_mmwave_d2d}
\textup{
Given the outage SINR threshold $\tau$, the coverage probability of a mmWave D2D link is given by
\begin{equation}
\label{eq:coverage_mmwave_d2d}
p_\text{c,D}(\tau) = 2\pi\lambda_\text{r} c \int_{x>0} x e^{ -\Lambda_\text{u}^\text{(N)}(\tau, x) - \Lambda_\text{u}^\text{(F)}(\tau, x) -\Lambda(x,\lambda_\text{r}) - \beta x } \mathrm{d}x,
\end{equation}
where $\lambda_\text{r}$ is the intensity of candidate relay UEs, $\Lambda(x,\lambda_\text{r})$ is defined in \eqref{eq:lambda_big_generic}, and $\Lambda_\text{u}^{\text{(N)}}(\tau, x)$ and $\Lambda_\text{u}^{\text{(F)}}(\tau, x)$ are defined as:
\begin{equation*}
\begin{split}
\Lambda_\text{u}^{\text{(N)}}(\tau, d_0) =& \frac{ \rho \lambda_\text{b} c}{\beta^2} \Big( \phi_\text{u} \big( 1 - \beta D_3 e^{-\beta D_3} - e^{-\beta D_3} \big) + (2\pi -\phi_\text{u}) \big( 1 - \beta D_1 e^{-\beta D_1} -e^{-\beta D_1} \big) \Big), \\
\Lambda_\text{u}^{\text{(F)}}(\tau, d_0) =& \frac{ \phi_\text{b} \rho \lambda_\text{b} c}{2\pi \beta^2} \Big( \phi_\text{u} \big( \beta D_3 e^{-\beta D_3} - \beta D_4 e^{-\beta D_4} + e^{-\beta D_3} - e^{-\beta D_4} \big) \\
& + (2\pi -\phi_\text{u}) \big( \beta D_1 e^{-\beta D_1} - \beta D_2 e^{-\beta D_2} + e^{-\beta D_1} - e^{-\beta D_2} \big) \Big),\\
\end{split}
\end{equation*}
in which $D_k$ is shorthand for $D_k(\tau, x)$ as defined in \eqref{eq:radius_ir}, but using the array gains $G_k$ of the D2D links.
}
\end{theorem}

\subsection{Coverage Probability of Noise-Limited MmWave Links}
\label{sec:noise_limited}

Recent research suggests that mmWave networks are more likely to be noise-limited than interference-limited due to the blockage effect \cite{Rangan:mmwave_noise_limited,Akdeniz:mmwave_evaluation}. With the noise-limited assumption, a mmWave link with distance $d_0$ experiences an outage if the signal-to-noise ratio (SNR) at the receiver is below a given SNR outage threshold $\tau$. Therefore, the coverage probability is
\begin{equation*}
\begin{split}
p_\text{c}(\tau) = \mathbb{P} \left[ \frac{g_0 A^{-1} d_0^{-\alpha}}{\sigma^2} > \tau \right] 
= \mathbb{P} \left[ d_0 < \left( \frac{g_0}{\tau A \sigma^2} \right)^{1/\alpha} \right] 
= F_{d_0}(x)\big|_{x=\left( \frac{g_0}{\tau A \sigma^2} \right)^{1/\alpha}},
\end{split}
\end{equation*}
where $F_{d_0}(x)$ is the CDF of $d_0$. In other words, the coverage probability of a noise-limited mmWave link is the probability that at least one transmitter (BS in cellular mode or candidate relay UE in D2D relay mode) with LOS path to the typical UE falls into the disc $\boldsymbol{B}(\boldsymbol{o}, \left( \frac{g_0}{\tau A \sigma^2} \right)^{1/\alpha})$. We have the following result for a noise-limited mmWave link.
\begin{corr}[Noise-limited mmWave coverage probability]
\label{corr:cov_mmwave_d2d}
\textup{
Given an SINR threshold $\tau$, the coverage probability of a noise-limited mmWave link is
\begin{equation}
\label{eq:cov_mmwave_noise_limited}
\begin{split}
p_\text{c}(\tau) = 1 - \exp \left(-\frac{2 \pi \lambda c}{\beta^2} \left( 1 - e^{-\beta (\frac{g_0}{\tau A \sigma^2})^{1/\alpha}} -\beta (\frac{g_0}{\tau A \sigma^2})^{1/\alpha} e^{-\beta (\frac{g_0}{\tau A \sigma^2})^{1/\alpha}} \right) \right),
\end{split}
\end{equation}
where $\lambda$ is the transmitter intensity, i.e., $\lambda = \lambda_\text{b}$ for a cellular link and $\lambda = \lambda_\text{r}$ for a D2D link.
}
\end{corr}

We note that the coverage probability expression for mmWave links is largely simplified under the noise-limited assumption. In Section~\ref{sec:simulation}, we evaluate the accuracy of this assumption.

\subsection{Coverage Probability of Microwave D2D Links}
\label{sec:coverage_d2d_microwave}

As in Section~\ref{sec:coverage_d2d_mmwave}, we approximate interferers on the microwave D2D link as a homogeneous PPP with intensity $\rho \lambda_\text{b}$, where $\rho$ is the multiplexing factor. Under the UE association rule introduced in Section~\ref{sec:ue_association}, we have the following result for microwave D2D links. 
\begin{theorem}[Microwave D2D coverage probability]
\label{theorem:cov_microwave_d2d}
\textup{
Given the outage threshold $\tau$ and candidate relay UE intensity $\lambda_\text{r}$, the coverage probability of a microwave D2D link, $p_{c,D}(\tau)$, is given by
\vspace{-8pt}
\begin{equation}
\label{eq:coverage_microwave_d2d}
p_\text{c,D}(\tau) = S_\text{N} p_\text{c,N}(\tau) + S_\text{L} p_\text{c,L}(\tau),
\vspace{-8pt}
\end{equation}
where 
\begin{equation}
S_\text{L} = 2\pi\lambda_\text{r} c \int_0^\infty x e^{-\Lambda(x,\lambda_\text{r}) -\beta x -\pi \lambda_\text{r} \tilde{a}^2 x^{2 \frac{\alpha_\text{L}}{\alpha_\text{N}}} +\Lambda( \tilde{a} x^{\frac{\alpha_\text{L}}{\alpha_\text{N}}}, \lambda_\text{r})} \mathrm{d}x
\end{equation}
is the probability that the typical UE associates with a LOS candidate relay UE, $S_\text{N} = 1 - S_\text{L}$ is the probability that the typical UE associates with a NLOS candidate relay UE, and $\tilde{a} = (A_\text{L}/A_\text{N})^{1/\alpha_\text{N}}$. 
In~\eqref{eq:coverage_microwave_d2d}, $p_\text{c,N}(\tau)$ and $p_\text{c,L}(\tau)$ are the coverage probabilities given that the nearest NLOS and nearest LOS candidate relay UEs are selected as a relay, respectively: i.e., 
\begin{equation}
\label{eq:coverage_micro_nlos}
\begin{split}
p_\text{c,N}(\tau) &= 2\pi\lambda_\text{r} \int_0^\infty x (1 -c e^{-\beta x}) e^{\Lambda(x,\lambda_\text{r}) -\pi\lambda_\text{r} x^2 - \mu\tau\sigma^2 A_\text{N} x^{\alpha_\text{N}}} \mathcal{L}_{I}(\mu\tau A_\text{N} x^{\alpha_\text{N}}) \mathrm{d}x \quad \text{and}\\
\end{split}
\end{equation}
\singlespacing
\vspace{-0.8cm}
\begin{equation}
\label{eq:coverage_micro_los}
\begin{split}
p_\text{c,L}(\tau) = 2\pi \lambda_\text{r} c \int_{x>0} e^{- \Lambda(x, \lambda_\text{r}) -\beta x -\mu\tau \sigma^2 A_\text{L} x^{\alpha_\text{L}}} \mathcal{L}_{I}(\mu\tau A_\text{L} x^{\alpha_\text{L}}) \mathrm{d}x, \\
\end{split}
\end{equation}
where $\mathcal{L}_{I}(s)$ is the Laplace transform of the interference, 
\begin{equation}
\label{eq:laplace_interference}
\begin{split}
\mathcal{L}_{I}(s) = \exp \left( -\pi \rho \lambda_\text{b} s^{2/\alpha_\text{N}} \frac{2\pi/\alpha_\text{N}}{\sin(2\pi/\alpha_\text{N})} \right).
\end{split}
\vspace{-8pt}
\end{equation}
}
\end{theorem}
\begin{proof} The proof is provided in Appendix~\ref{app:proof_cov_microwave_d2d}.
\end{proof}

Since UEs with NLOS paths to the destination UE can serve as relays, we expect to achieve better coverage using microwave D2D relays than using mmWave D2D relays. We will verify this intuition in Section~\ref{sec:simulation}.

\section{Spectral Efficiency with D2D Relays}
\label{sec:spectral}
Spectral efficiency (SE) is defined as the bit rate per unit bandwidth per cell. In general, the SE can be calculated using Shannon's theorem: $\gamma\triangleq \log_2(1+\text{SINR})$ (bits/s/Hz). Given the cellular coverage probability $p_\text{c,C}(\tau)$, the average downlink SE can be evaluated as follows \cite{bai:coverage}:
\begin{equation}
\label{eq:cellular_se}
\gamma = \frac{1}{\ln2} \int_{t>0} p_\text{c,C}(t) \frac{1}{1+t} \mathrm{d}t.
\end{equation}

In D2D relay-assisted networks, however, \eqref{eq:cellular_se}
no longer holds. Intuitively, the downlink SE of a relay-assisted network depends on the \textit{mode selection strategy}, which determines whether a UE is served in direct or relay mode.
In this paper, we consider a mode selection strategy that the downlink transmission is switched from direct cellular to D2D relay mode if the cellular link's SINR is below a threshold $\tau$, and a candidate relay UE is available which has SINRs above $\tau$ on both the BS-to-relay and relay-to-UE links. Apparently, the SE of a relay-assisted mmWave cellular network depends on the \textit{relaying SINR threshold} $\tau$. We first derive the average SE of the cellular downlink conditioned on the SINR being above or below $\tau$. Then, using $p_c(\tau)$ defined in~\eqref{eq:cov_prob}, we derive the average downlink SE using the total expectation theorem. Finally, we quantify the amount of uplink resources used by D2D links to support the relay transmissions.

The average downlink cellular SE conditioned on $\text{SINR}>\tau$ can be expressed as follows:
\begin{equation*}
\begin{split}
\gamma_\text{C}(\tau) &= \mathbb{E}[\log_2(1+\text{SINR})|\text{SINR}>\tau] \\
&= \frac{\mathbb{E}[\boldsymbol{1}_{\left[\tau, \infty\right)} (\text{SINR}) \log_2(1+ \text{SINR})]}{\mathbb{P}(\text{SINR}>\tau)}\\
&= \frac{1}{\ln2} \left( \ln(1+\tau) + \frac{1}{p_\text{c,C}(\tau)} \int_{\tau}^{\infty} \frac{p_\text{c,C}(t)}{t+1} \mathrm{d}t \right),
\end{split}
\end{equation*}
where $\boldsymbol{1}_A(x)$ is an indicator function that evaluates to 1 if $x \in A$ and 0 otherwise. Similarly, the average downlink cellular SE conditioned on $\text{SINR} \leq \tau$ can be expressed as:
\begin{equation*}
\gamma_{\bar{\text{C}}}(\tau) = \frac{1}{\ln2( 1- p_\text{c,C}(\tau))}\left( -p_\text{c,C}(\tau)\ln(1+\tau)+ \int_0^\tau \frac{p_\text{c,C}(x)}{1+x} \mathrm{d}x \right).
\end{equation*}
Then, by the total expectation theorem, the average SE on the downlink of a D2D relay-assisted network with relaying SINR threshold $\tau$ is
\begin{equation}
\label{eq:overall_se}
\gamma(\tau) =p_\text{c}(\tau) \gamma_\text{C}(\tau)  +(1-p_\text{c}(\tau)) \gamma_{\bar{\text{C}}}(\tau).
\end{equation}
Note that if $\tau\to 0$, then relay transmissions are not needed because $p_\text{c,C}(0) = 1$. If $\tau\to\infty$, then no candidate relays exist because $p_\text{c,C}(\infty) =  p_\text{c,D}(\infty) = 0$. Therefore, in both cases, $\gamma(\tau)\to \gamma$, where $\gamma$ is the average cellular SE defined in~\eqref{eq:cellular_se}. It is easy to verify this intuition from \eqref{eq:overall_se}.

The overall coverage probability in terms of an outage threshold can be thought of as the fraction of cellular transmissions with SINRs above that outage threshold. Thus, given the relaying threshold $\tau$, $(1-p_\text{c,C}(\tau))p_\text{c,R}(\tau) w_\text{dl}$ downlink resources are used by BS-to-relay links, where $w_\text{dl}$ is the cellular downlink bandwidth. Since the traffic volumes on the BS-to-relay and relay-to-UE links must be the same, the amount of uplink resources required to support D2D relaying is
\begin{equation}
\label{eq:d2d_resource}
\Upsilon(\tau) 
= \frac{\gamma_\text{C}(\tau)}{\gamma_\text{D}(\tau)}(1-p_\text{c,C}(\tau))p_\text{c,C}(\tau)p_\text{c,D}(\tau) w_\text{dl},
\end{equation}
where $\gamma_\text{D}(\tau)$ is the average D2D link SE given relaying threshold $\tau$, which can be derived similarly to $\gamma_\text{C}(\tau)$.

\section{Numerical Results}
\label{sec:simulation}

In this section, we validate our analytical results against simulations based on 3GPP network evaluation methodologies, and explore the effect of different parameters on the coverage probabilities and spectral efficiency. We first validate the LOS probability model~\eqref{eq:los_prob_simplify} in Section~\ref{sec:los_validate}. In Section~\ref{sec:sim_setup}, we describe the simulation setup. In Section~\ref{sec:simulation_verify}, we validate the derived coverage probabilities for the cellular and D2D links. In Section~\ref{sec:simulation_noise_limited}, we evaluate the noise-limited assumption for mmWave cellular links. In Section~\ref{sec:simulation_relay}, we investigate the coverage improvement achieved by D2D relaying under different network configurations. Finally, in Section~\ref{sec:simulation_se}, we evaluate the downlink SE and the effect of D2D transmissions on the uplink resources.

\subsection{Cylindrical Obstacle Model Validation}
\label{sec:los_validate}

To simplify the exposition and analysis, we adopt a cylindrical obstacle model for both outdoor urban and indoor office environments. Intuitively, cylindrical obstacles are reasonable approximations of human bodies in indoor environments. 
To verify the accuracy of the cylindrical obstacle model in outdoor environments, we compare the LOS probabilities predicted by it and the rectangular obstacle model~\cite{bai:coverage} against the statistical LOS probability obtained using real building data for the ultra-dense Chicago area \cite{kulkarni:coverage} \cite{chicago_data}.

In the cylindrical/rectangular obstacle models, buildings are approximated by cylinders/rectangles with average cross-sectional area matching the average cross-sectional area of the actual buildings. The obstacle intensity is determined by the average obstacle area and the ratio of the area covered by the obstacles. For example, using the cylindrical model, the obstacle intensity is determined as $\lambda_\text{o} = \xi/(\pi \mathbb{E}[R^2])$, where $\xi \in (0,1)$ denotes the ratio of the area that is covered by obstacles and the random variable $R$ denotes the obstacle radius. For the rectangular obstacle model, we consider two cases: $\mathbb{E}[L]=2\mathbb{E}[W]$ and $\mathbb{E}[L]=\mathbb{E}[W]$, where $\mathbb{E}[W]$ and $\mathbb{E}[L]$ are the expected length and width of the rectangular obstacles, respectively. To evaluate the statistical LOS probability using real building data, we randomly drop nodes in the simulated area. LOS exists between two nodes if neither is inside of a building and there is no building between them. The real building data for the Chicago area and the corresponding LOS probabilities are shown in Fig.~\ref{fig:los_validate}(a) and Fig.~\ref{fig:los_validate}(b), respectively. We note that $\xi \approx 0.38$ in the simulated area.

\begin{figure}
\centering
  \subfloat[Building distribution]{\includegraphics[width=.4\textwidth]{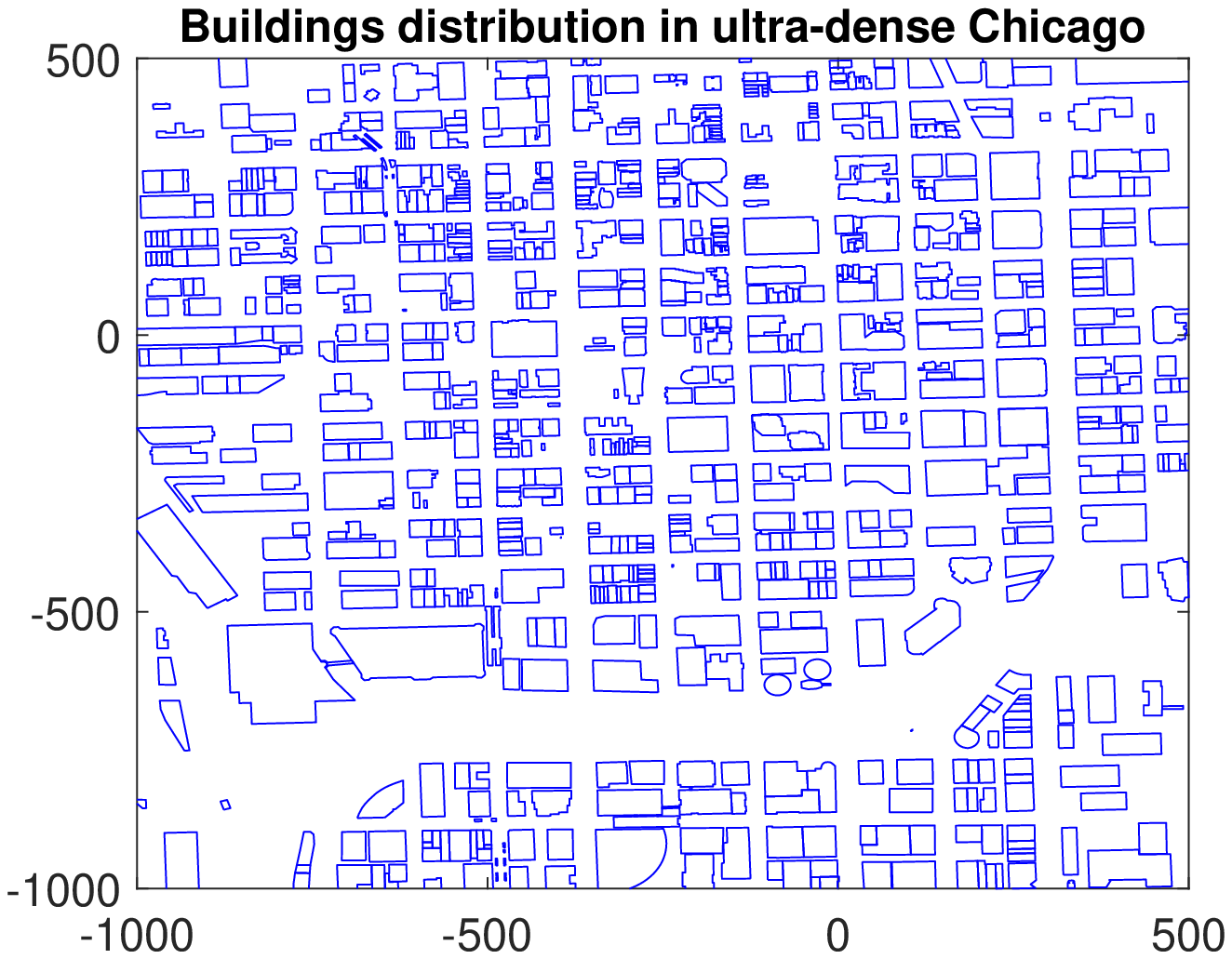}} 
  \subfloat[LOS probabilities]{\includegraphics[width=.4\textwidth]{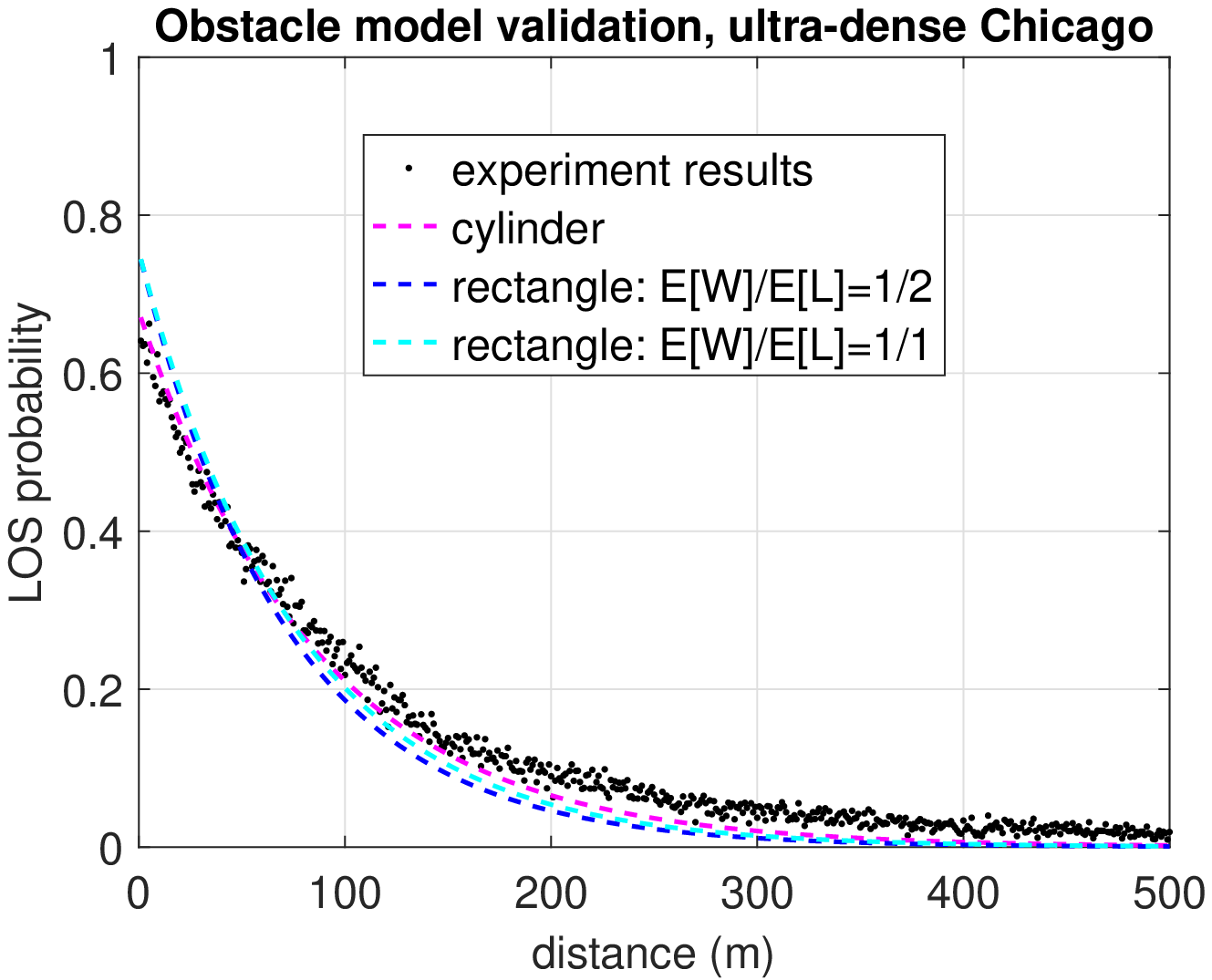}}   \caption{\label{fig:los_validate} Cylindrical obstacle model validation.}
\end{figure}

We can see that the cylindrical and rectangular obstacle models obtain very similar LOS probabilities, and the LOS probabilities obtained using the cylindrical obstacle model actually match the experimental results better than those obtained using the rectangular obstacle model. 

\subsection{Simulation Setup}
\label{sec:sim_setup}

We compare our analytical results against Urban Macro (UMa) and Indoor Office (Ind) 3GPP mmWave evaluation scenarios~\cite{3gpp:mmwave_TR}, and PPP-based network models. All of our simulations use the UE association strategy described in Section~\ref{sec:ue_association} to be consistent with our analytical results. Below, we describe our simulation setup. Key simulation parameters are listed in Table~\ref{table:simulation_parameter}.

\begin{table}
\small
\vspace{0.5in}
\caption{\label{table:simulation_parameter}List of configuration parameters for performance evaluation}
\centering 
\begin{tabular}{|l|l|l||l|l|l|}
\hline
\textbf{Parameters} & \textbf{UMa} & \textbf{Ind} & \textbf{Parameters} & \textbf{UMa} & \textbf{Ind} \\
\hline
BS intensity $\lambda_\text{b}$ (BS$/\text{m}^2$) & $4.62 \times 10^{-6}$ & $2 \times 10^{-3}$ & BS antenna height (m) & $25$ & $3$\\
\hline
Relay UEs per cell & $10$ & $3$ & BS Tx power (dBm) & $35$ & $24$\\
\hline
Obstacle radius $[r_\text{min}, r_\text{max}]$ & $[20, 30]$ & $[0.3, 0.6]$ & UE antenna height (m) & $1.5$ & $1$\\
\hline
Obstacle height $[h_\text{min}, h_\text{max}]$ & $[5, 25]$ & $[1, 2]$ & UE Tx power (dBm) & $23$ & $23$ \\
\hline
Obstacle thinning (cellular) $\eta_\text{C}$ & $0.5875$ & $0.5$ & Noise power (dBm/Hz) & $-174$ & $-174$\\
\hline
Obstacle thinning (D2D) $\eta_\text{D}$ & $1$ & $1$ & UE noise figure (dB) & $9$ & $9$ \\
\hline
\end{tabular}
\vspace{-0.2in}
\end{table}

\subsubsection{BS and UE distributions} 
In 3GPP's grid-based evaluation models, the UMa scenario has a hexagonal cell layout with an inter-site distance (ISD) of $500$ m. Meanwhile, the Ind scenario for mmWave evaluation has $12$ BSs in a $50$ m $\times$ $120$ m rectangular area, where each BS covers an area of size $25$ m $\times$ $20$ m (Figure $7.2$-$1$ in~\cite{3gpp:mmwave_TR}). For the PPP-based network models, the BSs are dropped according to 2D homogeneous PPPs with intensities $\lambda_\text{b}=4.62 \times 10^{-6}/\text{m}^2$ and $\lambda_\text{b}=2 \times 10^{-3}/\text{m}^2$ to match the average BS ISDs/densities in the 3GPP UMa and Ind scenarios, respectively. In both grid- and PPP-based network models, UEs are assumed to be dropped according to a 2D homogeneous PPP. If there are $n$ candidate relay UEs per cell on average, then the intensity of candidate relay UEs is $\lambda_\text{r} = n \lambda_\text{b}$, where $\lambda_\text{b}$ is the intensity of BSs. 

\subsubsection{Obstacle distributions} 
We assume that the obstacle radius $R$ and height $H$ are uniformly distributed on $[r_\text{min}, r_\text{max}]$ and $[h_\text{min}, h_\text{max}]$, respectively. Obstacle distribution parameters are set differently in the two scenarios according to their obstacle characteristics. Note that for the UMa scenario, we consider a moderately dense urban area with obstacle coverage $\xi=0.2$. In the Ind scenario, the obstacle parameters in Table~\ref{table:simulation_parameter} imply that there are $15$ obstacles in a $100~\text{m}^2$ area on average. The obstacle thinning parameters $\eta_\text{C}$ and $\eta_\text{D}$ are computed from the distribution of $H$ and the BS/UE antenna heights.

\subsubsection{Antenna configurations} 
To model beamforming, we assume that mmWave transmitters and receivers are equipped with uniform planar antenna arrays. We consider BSs with $N_\text{b} \times N_\text{b}$ antenna arrays, with $N_\text{b}=8$ and $N_\text{b}=4$, i.e., $64$ and $16$ transmit antennas, respectively. We assume that all UEs have $N_\text{u}\times N_\text{u}=2 \times 2$ antenna arrays for both transmitting and receiving. Antenna array gains and beam widths are determined according to \eqref{eq:antenna_parameter}. For example, the main lobe gain, side lobe gain, and beam width of a $4 \times 4$ square antenna array at the BS are $Gm_\text{b}=12.04$ dB, $Gs_\text{b}=0.69$ dB and $\phi_\text{b}=24.8^{\circ}$, respectively. We assume that microwave D2D links use a single transmit and receive antenna.

\subsubsection{Channel models} 
We use path loss models from 3GPP evaluation methodology documents. For mmWave links, the LOS path loss models for the UMa and Ind scenarios are \cite{3gpp:mmwave_TR}:

\singlespacing
\vspace{-0.8cm}
\begin{equation*}
\begin{split}
PL &=~ 32.4 + 20 \log_{10}(d)+ 20\log_{10}(f_c) \quad \text{and} \\
PL &=~ 32.4 + 17.3 \log_{10}(d)+ 20\log_{10}(f_c), \\
\end{split}
\end{equation*}
\doublespacing
respectively, where $d$ is the distance in meters and $f_c$ is the carrier frequency in GHz. For microwave D2D \cite{3gpp:d2d_843}, the LOS/NLOS path losses for the UMa scenario are

\singlespacing
\vspace{-0.8cm}
\begin{equation*}
\begin{split}
PL_\text{L} &=~ 27.0 + 22.7 \log_{10}(d)+ 20\log_{10}(f_c) \quad \text{and} \\
PL_\text{N} &=~ 14.78+5.83\log_{10}(1.5) + (44.9-6.55\log_{10}(1.5)) \log_{10}(d) + 34.97\log_{10}(f_c), \\
\end{split}
\end{equation*}
\doublespacing
respectively, and in the Ind scenario are

\singlespacing
\vspace{-0.8cm}
\begin{equation*}
\begin{split}
PL_\text{L} =& 89.5 + 16.9 \log_{10} (d/1000) \quad \text{and} \\
PL_\text{N} =& 147.4 + 43.3 \log_{10} (d/1000), \\
\end{split}
\end{equation*}
\doublespacing
respectively. In the simulations, we set the carrier frequencies for mmWave and microwave links to $28$ GHz and $2$ GHz, respectively, and the bandwidths to $100$ MHz and $20$ MHz, respectively. 
The linear scale path loss parameter $A$ in~\eqref{eq:signal_model_mm} can be computed from these path loss equations as described in Section~\ref{sec:channel_models}.
Note that 3GPP uses 3D distance in its mmWave path loss models. To avoid complicating the analysis in this paper, we use 2D distance for $d$, which is commonly assumed in stochastic geometry based network performance analysis.\footnote{The 2D and 3D distances are different in cellular links because of the different BS and UE antenna heights. Using 3D distance leads to a non-homogeneous point process of BSs from the perspective of the typical UE (because they are on different plains), which significantly complicates the analysis. On the other hand, using the 2D distance allows us to maintain homogeneity of the point process and keep the analysis relatively simple. To validate our decision to use 2D distance instead of 3D distance, we compare the path loss and coverage probabilities under each assumption. For UMa (resp. Ind), the path loss error at the $5^\text{th}$ and $50^\text{th}$ percentiles of desired BS-UE link distances are $0.52$ (resp. $1.41$) and $0.03$ (resp. $0.08$) dB, respectively. In the UMa scenario, the average percent error in the coverage probabilities is only 0.94\% (resp. 0.85\%) when there are 8 x 8 (resp. 4 x 4) BS antennas. On the other hand, for the Ind scenario, the average percent error in the coverage probabilities is only 5.38\% (resp. 6.11\%) when there are 8 x 8 (resp. 4 x 4) BS antennas. We conclude that approximating the 3D distance with the 2D distance does not significantly affect the path loss for the majority of mmWave cellular links and, furthermore, results in only minor underestimation of the coverage probabilities.}

To calculate the SINR in the mmWave simulations, we consider the interference power from all effective LOS interferers. In contrast, our analytical results only consider dominant interferers. The multiplexing factor $\rho$ is set to be $1$, which means we consider a fully loaded uplink network model with orthogonal multiple access such that the BS schedules exactly one cellular UE or D2D UE on each uplink sub-channel at a given time within a cell.

\subsection{Validation of Analytical Expressions}
\label{sec:simulation_verify}

In this section, we compare our analytical results against the UMa and Ind 3GPP mmWave performance evaluation scenarios, and PPP-based network simulations, as described in Section~\ref{sec:sim_setup}. 

\subsubsection{MmWave Cellular Link Coverage}

\begin{figure}
\centering
  \subfloat[3GPP UMa (Analy \& Grid Sims \& PPP Sims).]{\includegraphics[width=.4\textwidth]{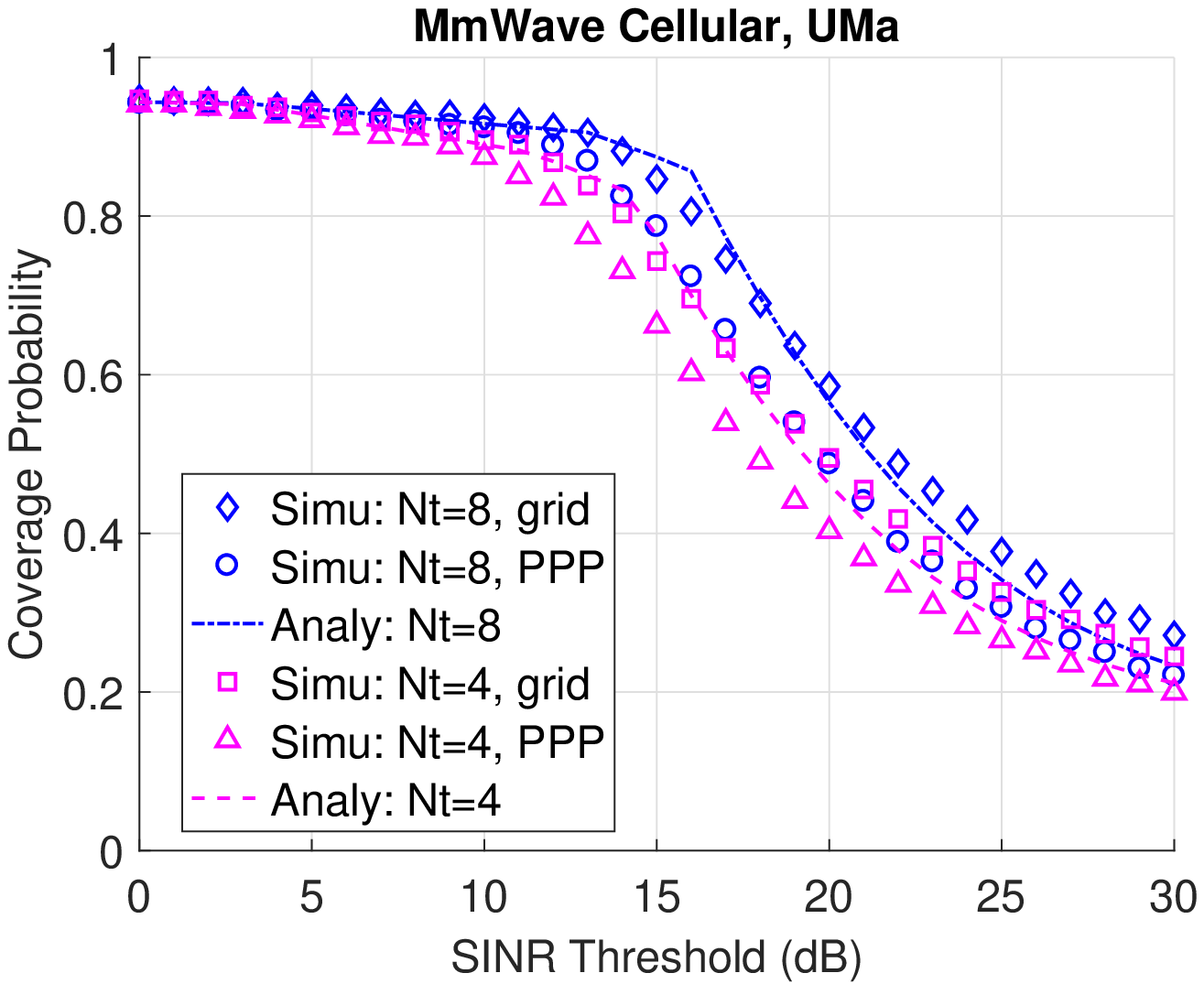}} 
  \subfloat[3GPP Ind (Analy \& Grid Sims \& PPP Sims).]{\includegraphics[width=.4\textwidth]{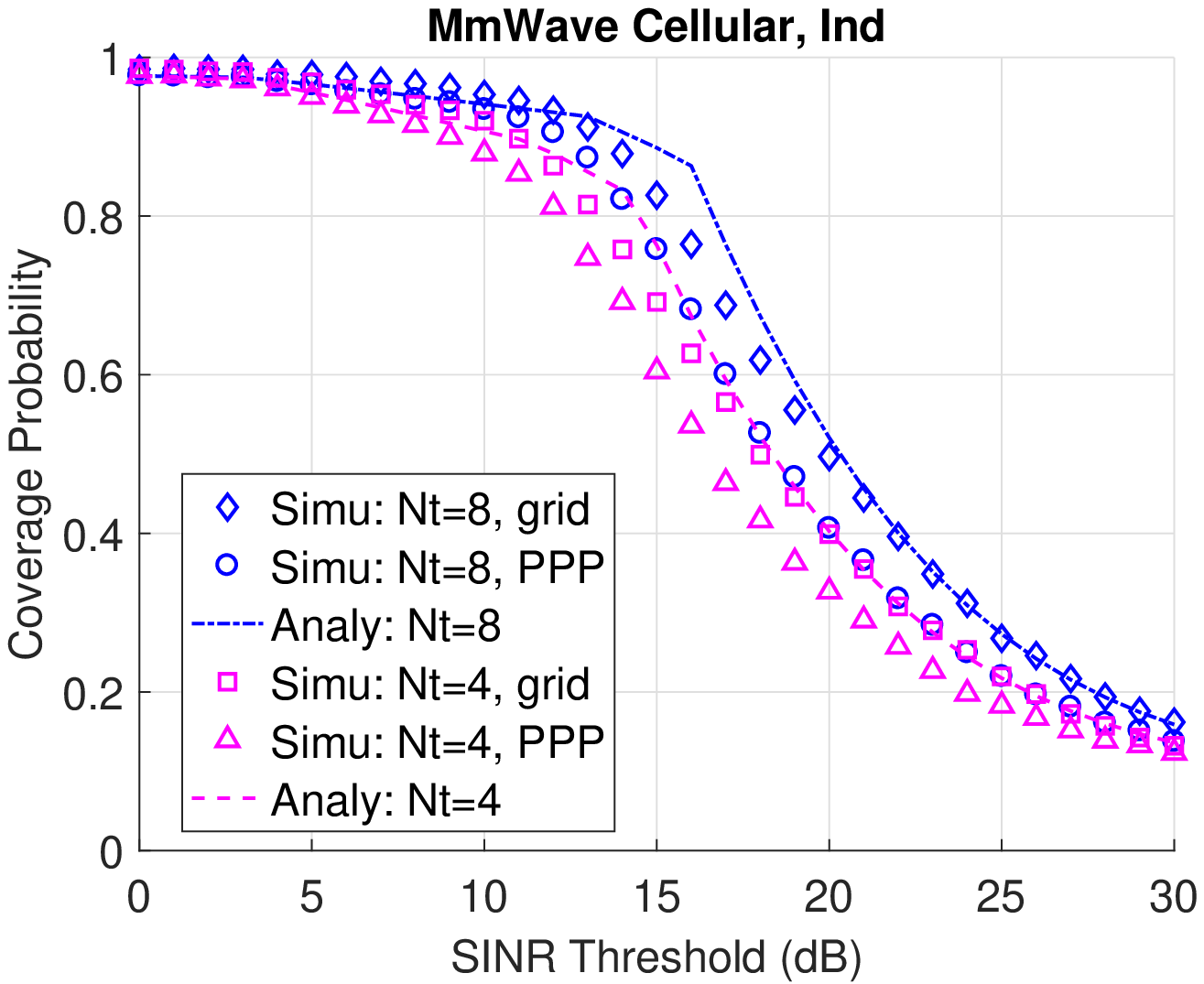}}   \caption{\label{fig:simu_verify_cellular} Validation of the mmWave cellular coverage probabilities in \eqref{eq:cov_mmwave_cellular} with respect to the SINR threshold.}
\end{figure}

In Fig.~\ref{fig:simu_verify_cellular}, we compare the mmWave cellular coverage probabilities, obtained analytically from \eqref{eq:cov_mmwave_cellular}, against PPP- and 3GPP grid-based model simulations. For illustration, the obstacle densities, $\xi$, for the UMa and Ind scenarios are set to $0.2$ and $0.08$, respectively. 
As expected, the analytical results obtained by dominant interferer analysis provide an an upper bound on the PPP-based model simulations. Additionally, for large SINR thresholds, the analytical results for the UMa scenario are bounded by the grid and PPP simulations. We conclude that the analytical coverage probability expression in~\eqref{eq:cov_mmwave_cellular} derived by dominant interferer analysis is a reasonable approximation in both 3GPP UMa and Ind scenarios. 

Another interesting observations about Fig.~\ref{fig:simu_verify_cellular} is, the analytical coverage probability declines faster starting around $\tau = 16$ dB. This happens because, if $\tau < 16$ dB, then the radius of the off-boresight NIR $D_1(\tau,d_0)$ defined in~\eqref{eq:radius_ir} and illustrated in Fig.~\ref{fig:interference_region_cellular} is smaller than $d_0$. In other words, the NIR determined by $D_1(\tau,d_0)$ does not contribute interference when $\tau < 16$ dB. Note that the SINR threshold that the decrease of coverage probability becomes faster depends on antenna configurations and fading model. We can also observe that the BS's antenna array size has a significant effect on the coverage, with larger arrays achieving higher coverage probabilities. In the remaining evaluations, we default to $8\times 8$ BS antenna arrays.

In Fig.~\ref{fig:simu_verify_cell_ob}, we show the mmWave cellular coverage probabilities in \eqref{eq:cov_mmwave_cellular} with respect to the SINR threshold for several obstacle densities. We observe that the analytical results obtained by dominant interferer analysis become increasingly accurate as the obstacle density increases. This observation also implies that dominant interferer analysis cannot provide accurate SINR distribution when obstacles are less dense. For example, in Ind, dominant interferer analysis accurately models the SINR distribution when $\xi>0.12$. 
Another interesting observation is that higher obstacle densities lead to higher coverage probabilities at high SINR thresholds. In other words, blockages help UEs achieve higher SINR in mmWave networks because they block signals from potential interferers.

\begin{figure}
\centering
  \subfloat[3GPP UMa (Analy \& PPP Sims)]{\includegraphics[width=.4\textwidth]{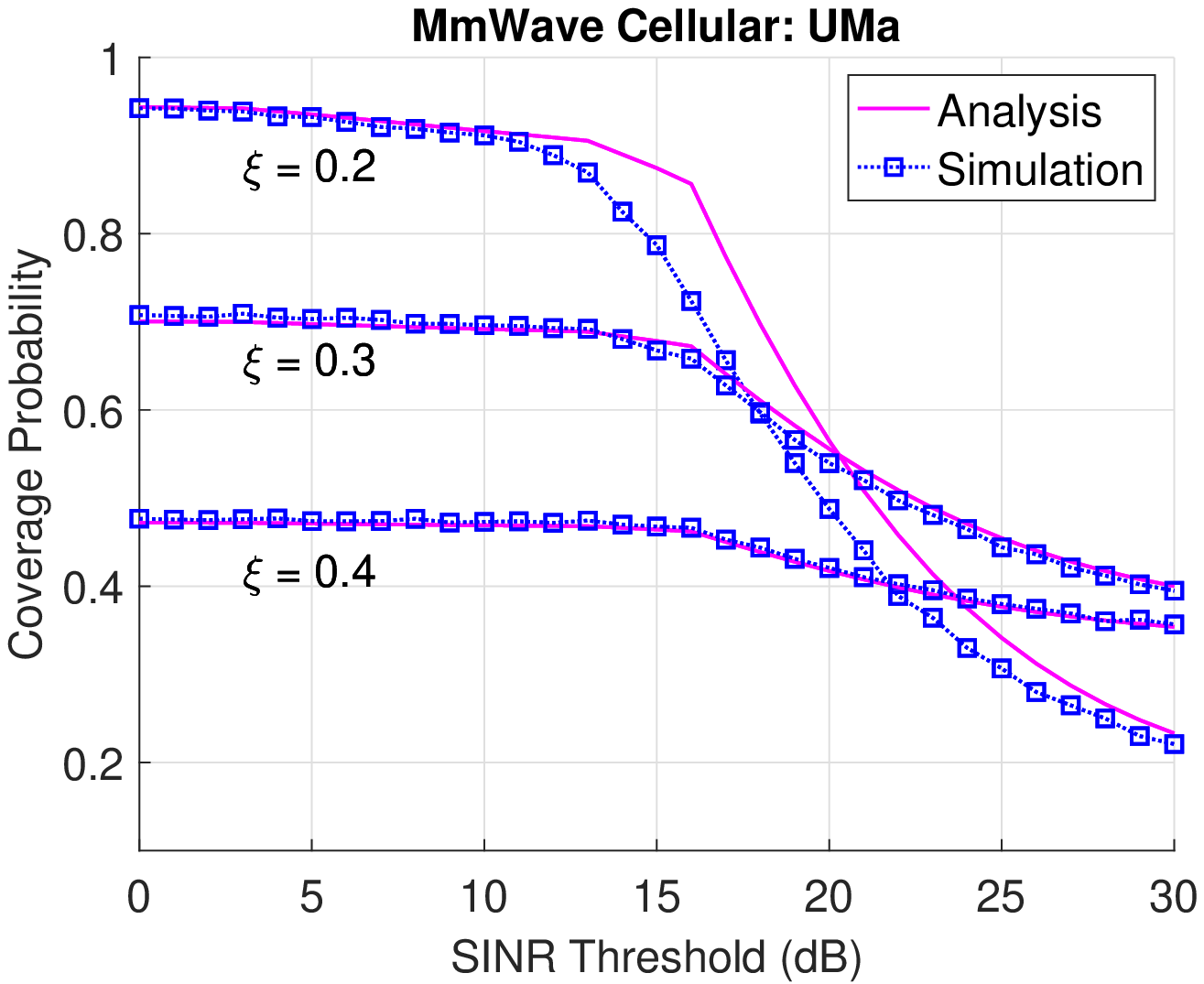}} 
  \subfloat[3GPP Ind (Analy \& PPP Sims)]{\includegraphics[width=.4\textwidth]{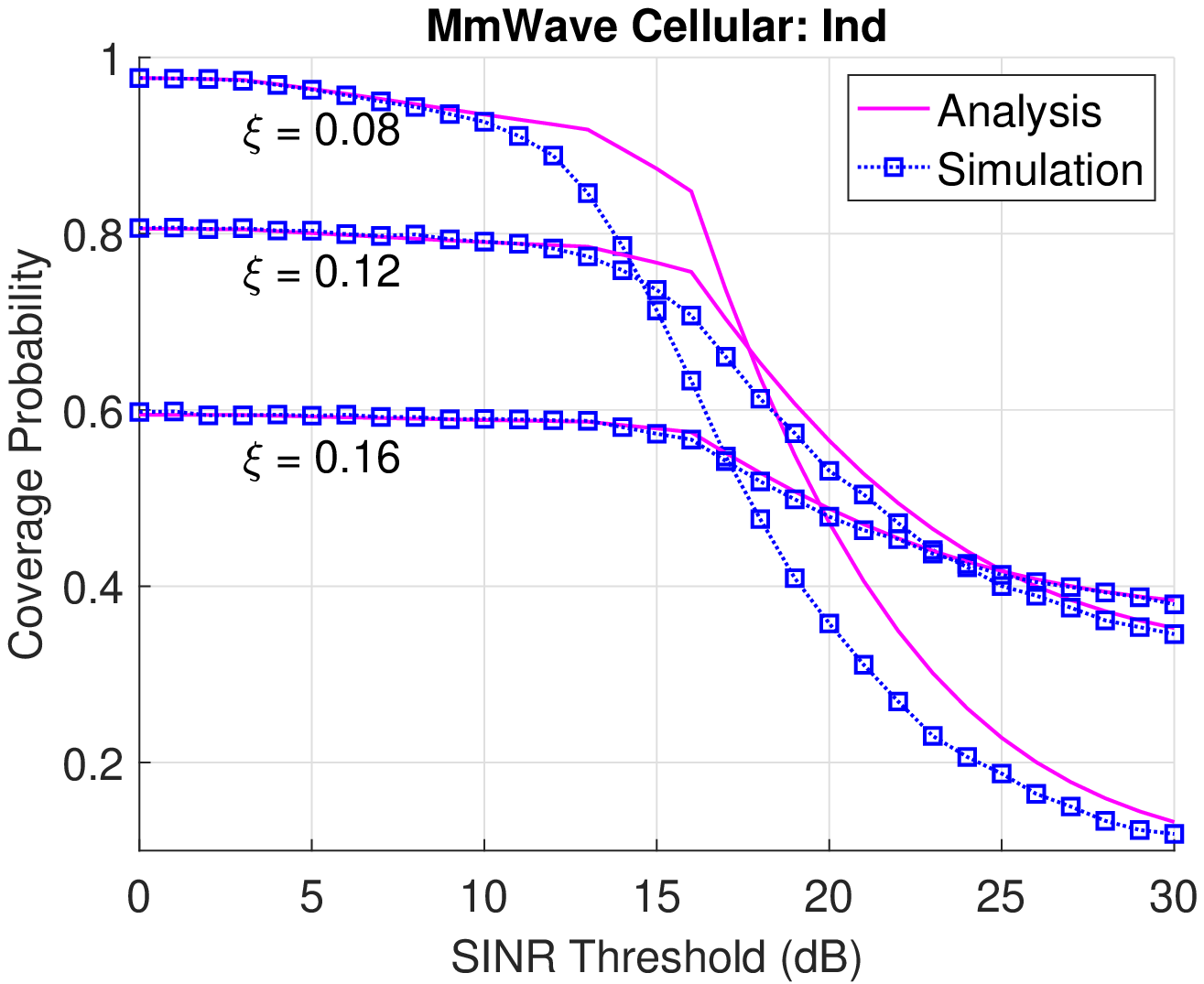}}   \caption{\label{fig:simu_verify_cell_ob} Validation of the mmWave cellular coverage probabilities \eqref{eq:cov_mmwave_cellular} with respect to obstacle densities against the PPP-based model simulations.}
\end{figure}

\subsubsection{D2D Link Coverage}

The analytical coverage probabilities for mmWave and microwave D2D links are shown in Fig.~\ref{fig:simu_verify_d2d} with respect to the SINR threshold and compared against PPP-based model simulations. We can see that dominant interferer analysis provides a significantly tighter upper bound on the PPP-based model simulations than for mmWave cellular links. This is because the relatively lower antenna heights of D2D transmitters and receivers result in more blockages and less interference from distant interferers. For microwave D2D links, the analytical results closely align with the PPP-based model simulations. Thus, we conclude that approximating uplink interferers using a homogeneous PPP is acceptable.

\begin{figure}
\centering
  \subfloat[3GPP UMa (Analy \& PPP Sims)]{\includegraphics[width=.4\textwidth]{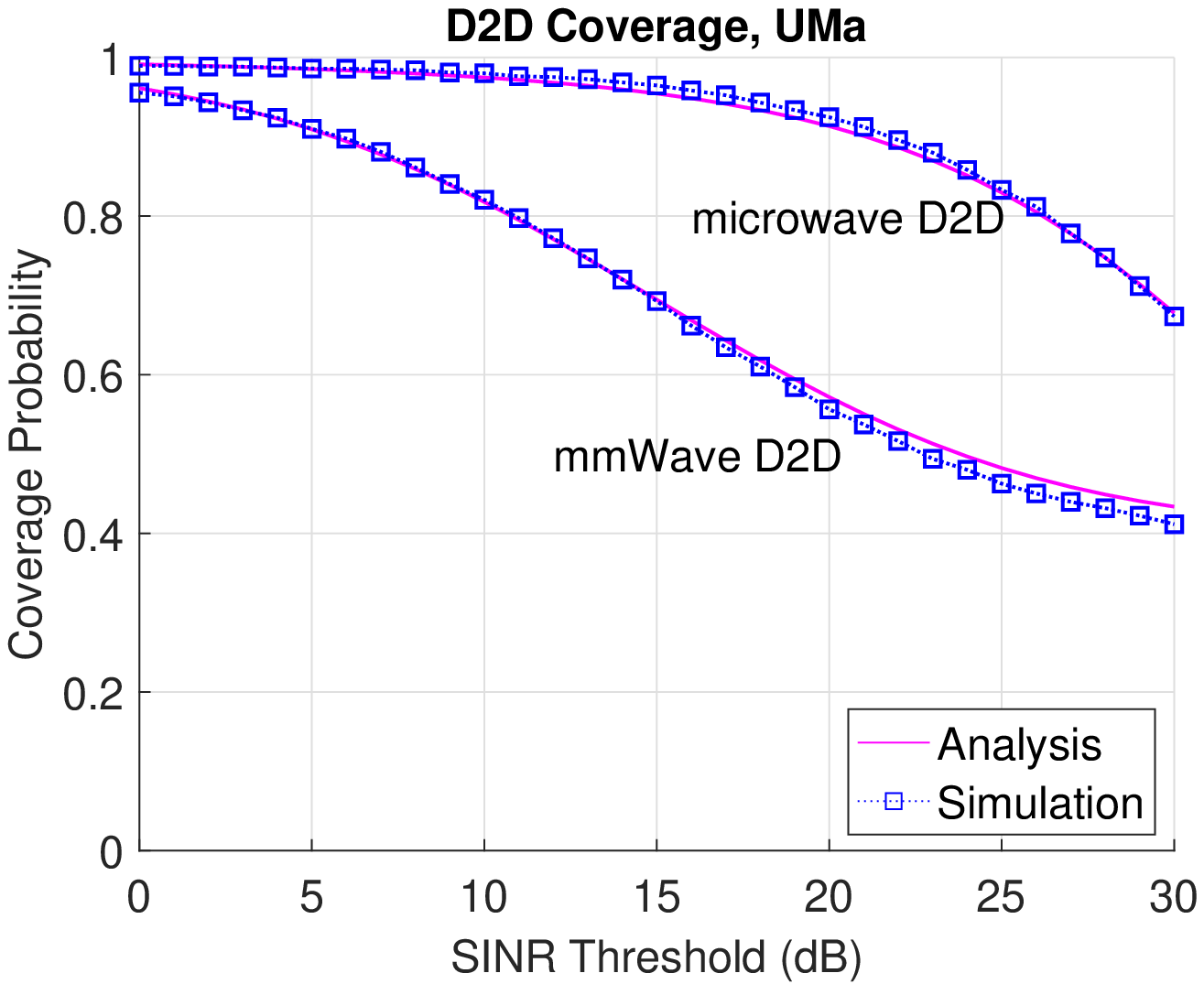}} 
  \subfloat[3GPP Ind (Analy \& PPP Sims)]{\includegraphics[width=.4\textwidth]{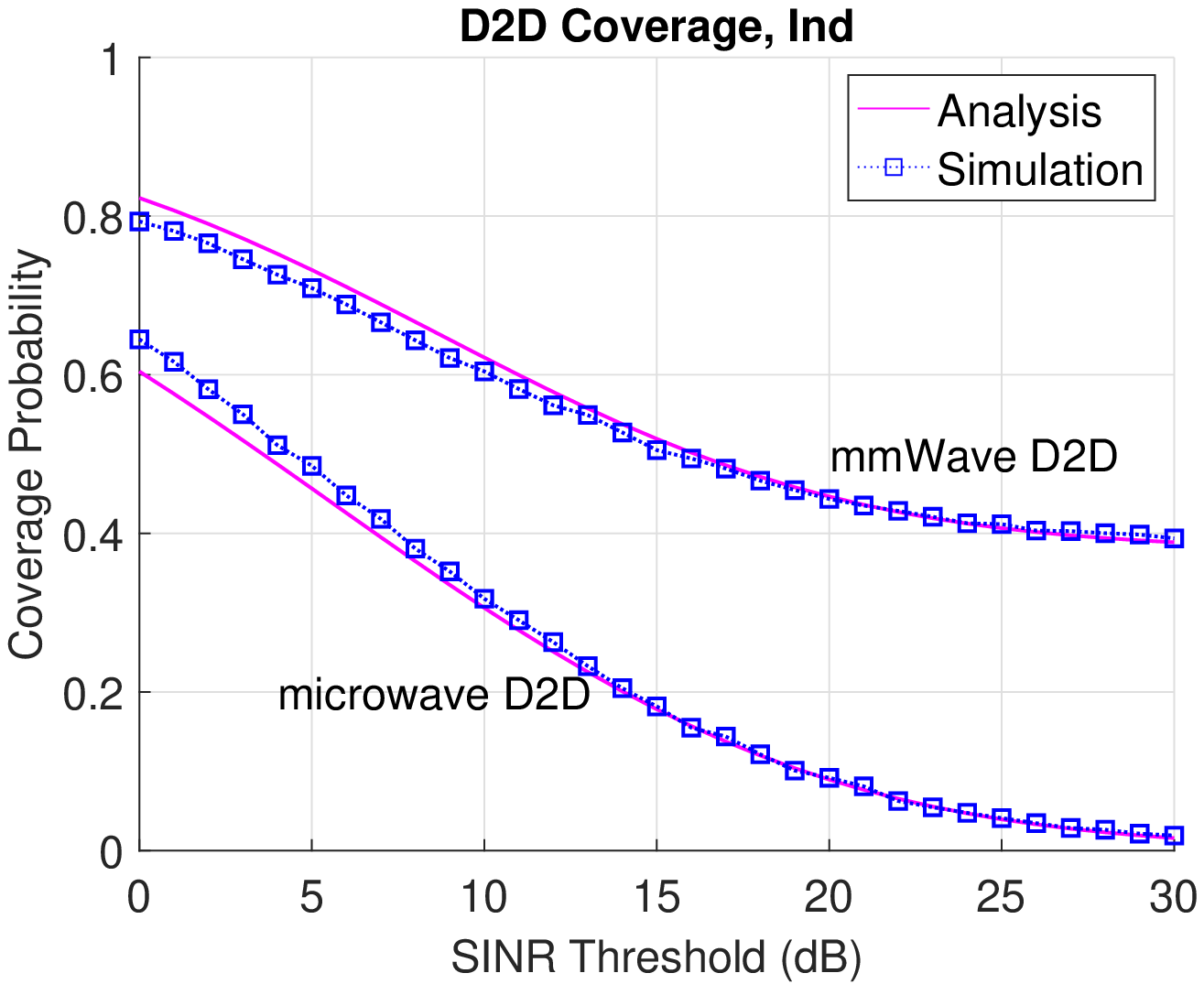}}  \caption{\label{fig:simu_verify_d2d} Validation of the mmWave and microwave D2D link coverage probabilities against simulations.}
\end{figure}

Comparing the two D2D options, we observe that microwave D2D performs better than mmWave D2D in the UMa scenario. This is because microwave D2D links can be established under NLOS conditions. However, microwave D2D links perform much worse than mmWave D2D links in the Ind scenario. This is because the dense BS deployment and fully utilized resources in each cell ($\rho=1$) cause severe interference in the uplink microwave spectrum. In contrast, mmWave D2D links experience less interference due to blockages and because the antenna arrays reject interference from off-boresight directions. Evidently, microwave D2D is a worse choice for extremely dense BS deployments, i.e. scenarios with dense interferers.

\subsection{Noise-Limited MmWave Link Coverage}
\label{sec:simulation_noise_limited}

In this section, we study the mmWave cellular link coverage probability under the noise-limited assumption. We compare the analytical results calculated according to \eqref{eq:cov_mmwave_noise_limited} with PPP-based BS model simulations to determine if we can ignore interference on mmWave links. In the simulations, both interference and noise are taken into account. 

The coverage probabilities for noise-limited mmWave cellular links, obtained analytically from~\eqref{eq:cov_mmwave_noise_limited}, are shown in Fig.~\ref{fig:noise_limited_cellular} with respect to the BS intensity. The noise-limited analytical results are compared against the analytical results obtained by dominant interferer analysis and PPP-based BS model simulations considering both interference and noise. Note that the urban macro and indoor office scenarios are based on the UMa and Ind scenarios, respectively, with parameter configurations in Table~\ref{table:simulation_parameter}, but with variations in the BS and obstacle densities. Here, we fix the outage threshold to $\tau=10$ dB and express the BS density using the equivalent ISD.

\begin{figure}
\centering
  \subfloat[Urban Macro (Analy \& PPP Sims)]{\includegraphics[width=.4\textwidth]{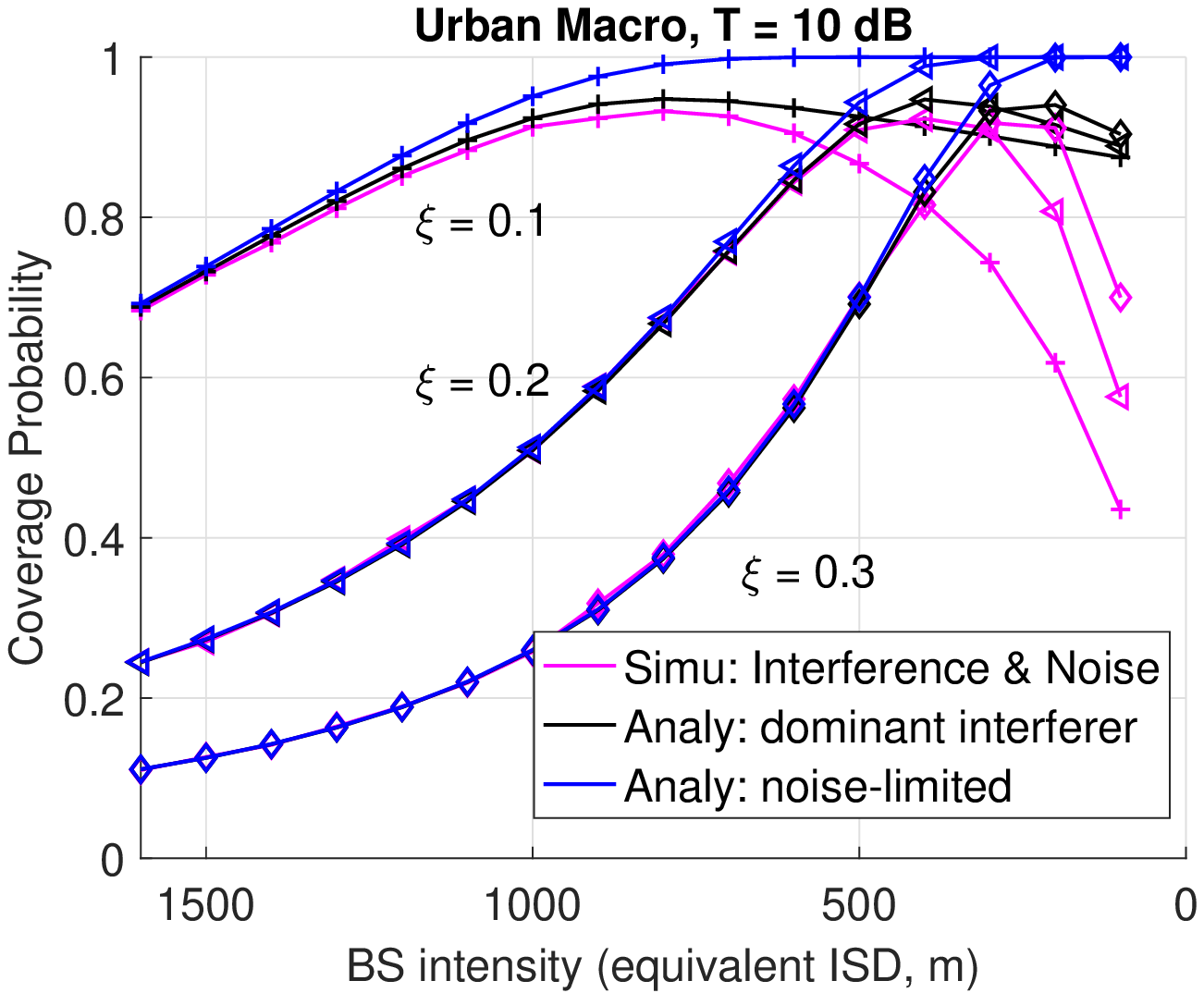}}
  \subfloat[Indoor Office (Analy \& PPP Sims)]{\includegraphics[width=.4\textwidth]{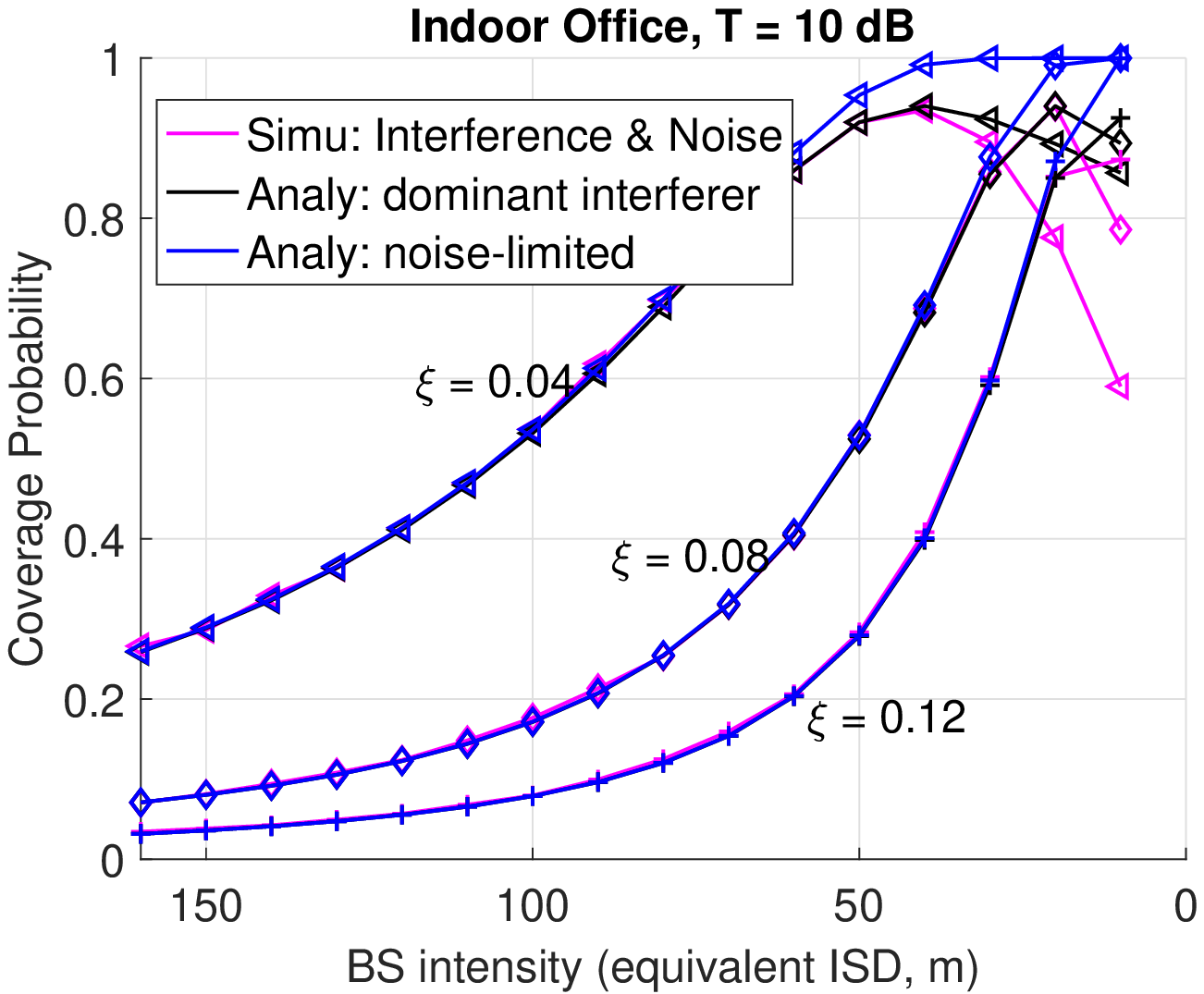}}  
\caption{\label{fig:noise_limited_cellular} Validation of the coverage expression in \eqref{eq:cov_mmwave_noise_limited} for noise-limited mmWave cellular links against the PPP-based model simulations. Note that there is an optimal BS density that increases with the obstacle density.}
\end{figure}

In both scenarios, the mmWave cellular link is noise-limited at low BS densities/large ISDs and interference-limited at high BS densities/small ISDs. We observe that for each obstacle density, there is an effective ISD \textit{threshold} below which the network tends to be interference-limited. For example, for $\xi=0.2$ in the urban macro scenario and $\xi=0.08$ in the indoor office scenario, ISDs below 700 m and 40 m, respectively, demonstrate interference-limited behavior. Furthermore, this BS density threshold increases with the obstacle density. We also observe that dominant interferer analysis becomes inaccurate at extreme BS densities/very small ISDs due to the increase in interference in extremely dense mmWave BS deployments.

\subsection{Coverage Improvement Enabled by D2D Relaying in MmWave Cellular Networks}
\label{sec:simulation_relay}

Now that the coverage probability expressions for the cellular and D2D links have been validated, we can calculate the coverage probabilities in D2D relay-assisted mmWave cellular networks using~\eqref{eq:cov_prob}.

\subsubsection{Coverage Improvement with 3GPP Configurations}

The coverage probabilities for a D2D relay-assisted mmWave cellular network, obtained analytically by \eqref{eq:cov_prob}, are shown in Fig.~\ref{fig:cov_comparison_3gpp} with respect to the SINR threshold. We note that D2D relays improve the coverage in both UMa and Ind scenarios, with larger gains in the UMa scenario due to the larger average transmit distances.

\begin{figure}
\centering
  \subfloat[3GPP UMa (Analytical)]{\includegraphics[width=.4\textwidth]{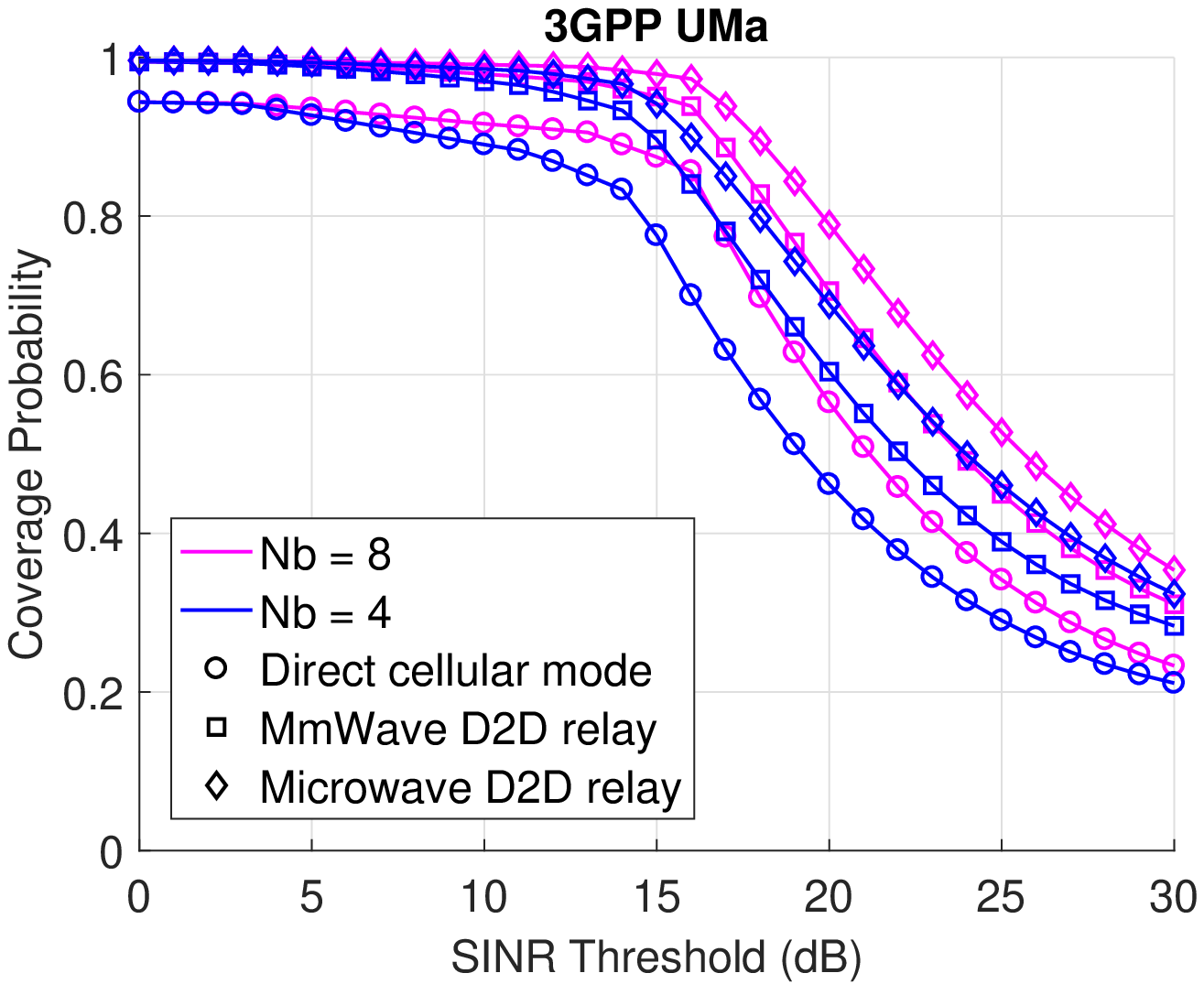}}  
  \subfloat[3GPP Ind (Analytical)]{\includegraphics[width=.4\textwidth]{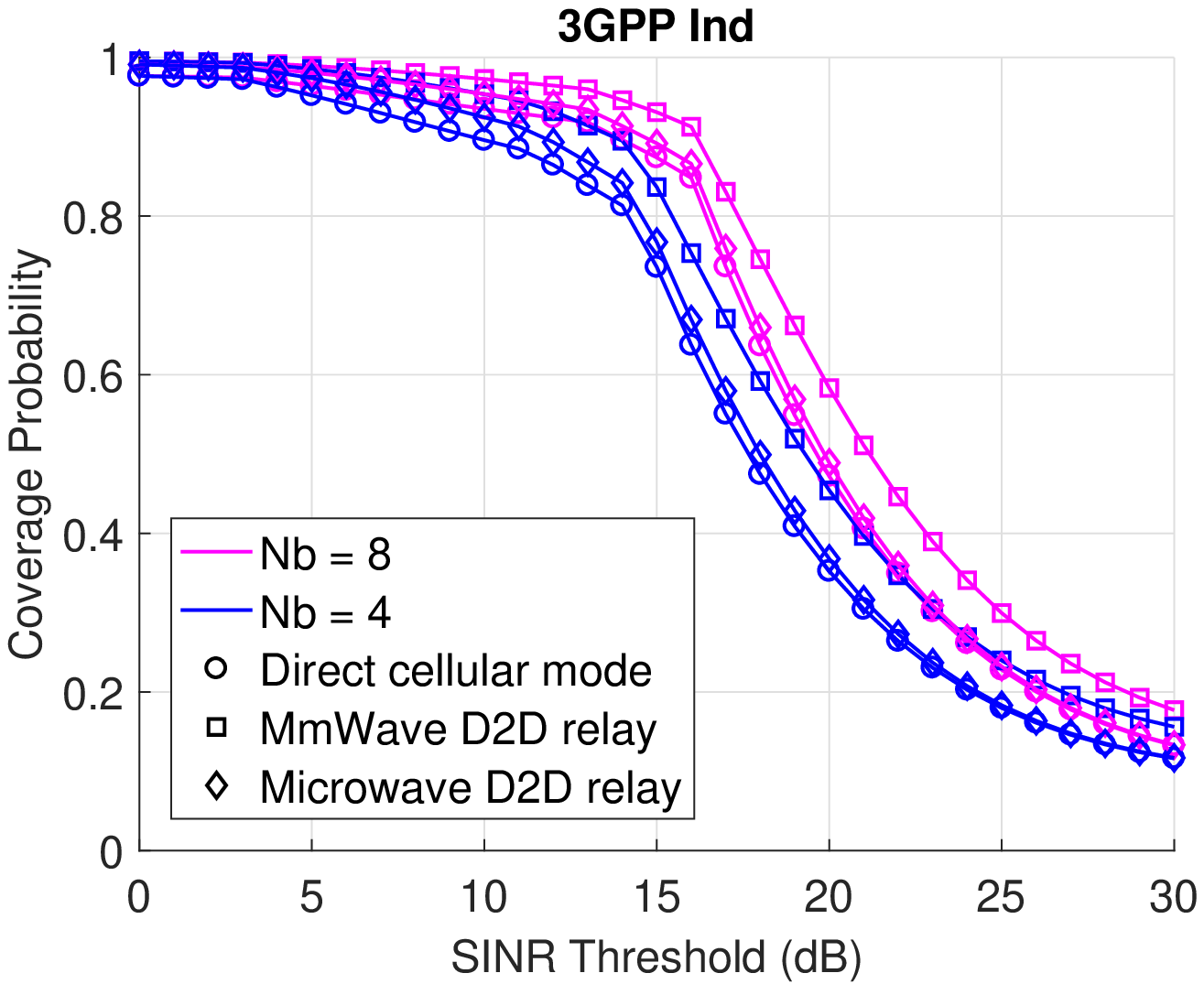}}
\caption{\label{fig:cov_comparison_3gpp} Coverage improvement of a D2D relay-assisted mmWave cellular network.}
\end{figure}

\subsubsection{Coverage vs. BS/Obstacle Density} 
The coverage probabilities for a D2D relay-assisted mmWave cellular network, obtained analytically as above, are shown in Fig.~\ref{fig:simu_bs_density} with respect to the BS intensity. 
Note that we fix candidate relay UE intensities in a given scenario. We observe that D2D relays improve the coverage for all BS densities for the selected SINR threshold $\tau = 10$ dB. We also observe that, given a specific UE intensity, there is an optimal BS deployment density that increases with the obstacle cover ratio.

\begin{figure}
\centering
  \subfloat[Urban Macro (Analytical)]{\includegraphics[width=.4\textwidth]{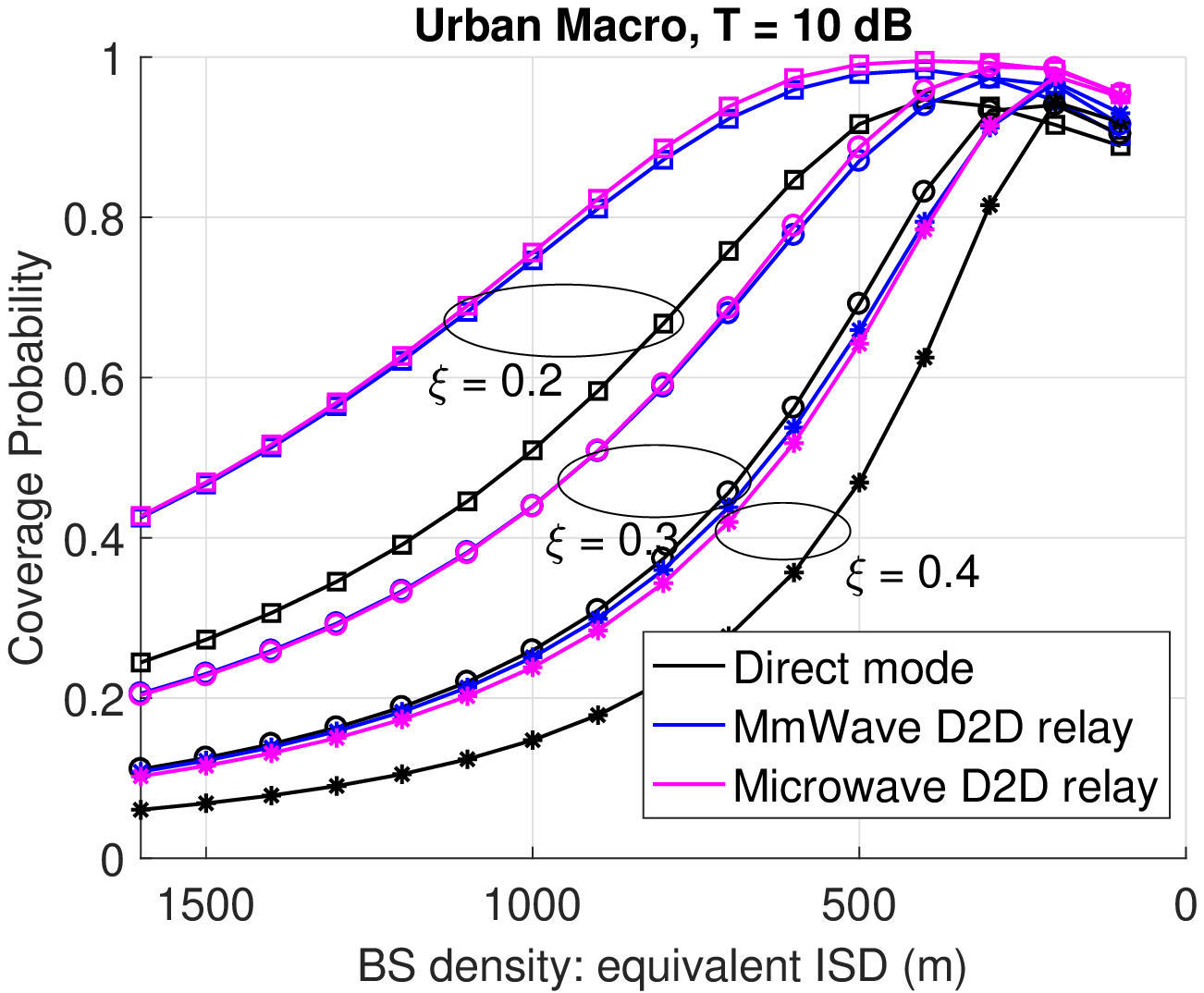}} 
  \subfloat[Indoor Office (Analytical)]{\includegraphics[width=.4\textwidth]{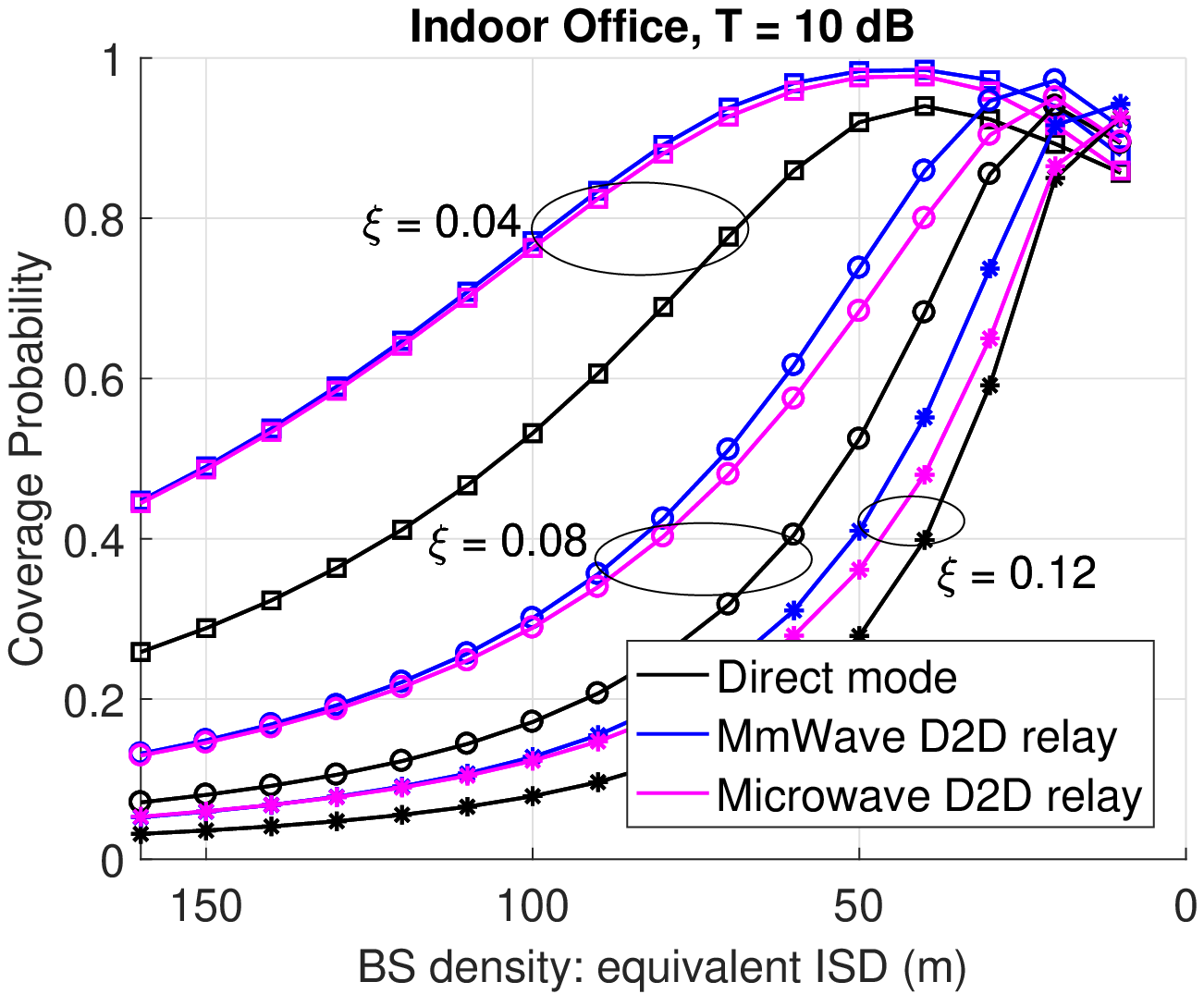}}  \caption{\label{fig:simu_bs_density} Coverage improvement in a relay-assisted mmWave cellular network. Note that there is an optimal BS deployment density that increases with the obstacle coverage ratio.}
\end{figure}

\subsection{Downlink Spectral Efficiency and the Effect of D2D on the Cellular Uplink}
\label{sec:simulation_se}
The SE is defined in \eqref{eq:cellular_se} and \eqref{eq:overall_se} and is shown in Fig.~\ref{fig:average_se}(a). We can see that the downlink SE is improved by both mmWave D2D relays and microwave D2D relays in the UMa and Ind scenarios compared to the cellular downlink without relaying. 
Moreover, it is clear that the SE improvement depends on the relaying threshold. For example, for microwave D2D relays in the UMa scenario, the SE is maximized when $\tau=21$ dB, which corresponds to a $16.3\%$ improvement over the cellular downlink without relaying. Note that, in a wireless communication system, the maximum exploitable SINR $\tau_\text{max}$ is limited by the system implementation, e.g., the available modulation schemes. We set $\tau_\text{max}=40$ dB.

Fig.~\ref{fig:average_se}(b) shows the uplink resources used for D2D transmissions, normalized by the uplink bandwidth $w_\text{ul}$, i.e., $\Upsilon(\tau)/w_\text{ul}$, where $\Upsilon(\tau)$ is defined in~\eqref{eq:d2d_resource}. We can see that the uplink resource fraction required to support D2D relaying is less for mmWave D2D links than for microwave D2D links. This is due to the difference in bandwidths available in the mmWave and microwave spectrum. In the UMa scenario, more than $100\%$ of microwave uplink spectrum is required to support the SE shown in Fig.~\ref{fig:average_se}(a) when $\tau$ is approximately $21$ dB. This is due to (i) the asymmetry in the bandwidths available on the mmWave cellular and the microwave D2D links, which are required to carry the same traffic volume, and (ii) the fact that D2D relay transmissions are more frequent over microwave D2D links because of their higher coverage probabilities in the UMa scenario. Unfortunately, this means that the SE improvement shown in Fig.~\ref{fig:average_se}(a) for microwave D2D relays in the UMa scenario is not achievable at some relaying thresholds. We conclude that mmWave uplink spectrum is better suited for carrying the D2D traffic.

Lastly, we would like to highlight the fact that our SINR threshold-based mode selection strategy is designed for coverage enhancement rather than spectral efficiency enhancement. As such, we believe that the combined coverage and spectral efficiency improvements afforded by mmWave D2D relaying outweigh the cost of the required uplink resources.

\begin{figure}
\centering
  \subfloat[Average Spectral Efficiency]{\includegraphics[width=.4\textwidth]{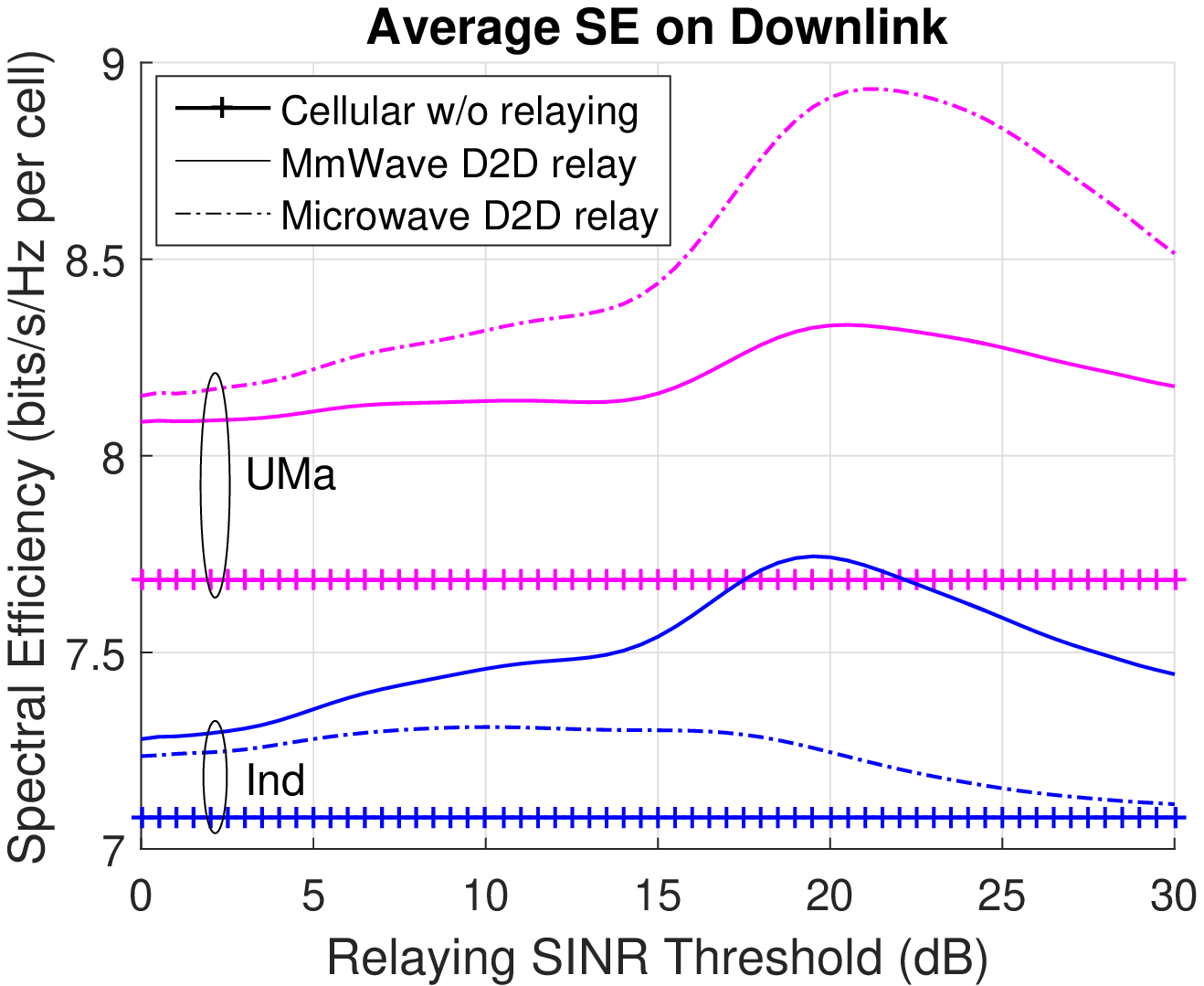}} 
  \subfloat[D2D link resource]{\includegraphics[width=.4\textwidth]{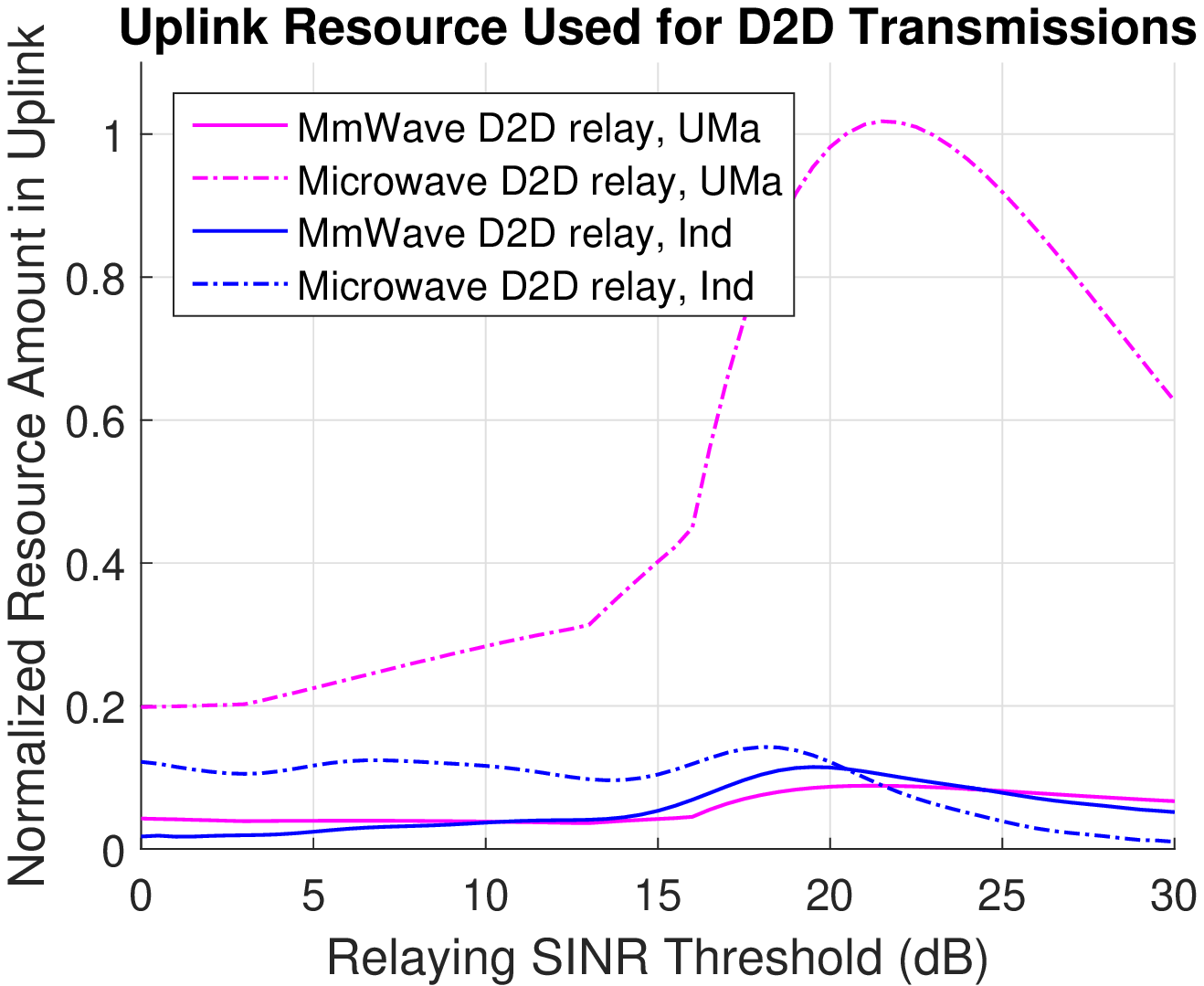}}
\caption{\label{fig:average_se} Downlink SE improvement and uplink resource used for D2D transmissions.}
\end{figure}

\section{Conclusion}
\label{sec:conclusion}

We envision mmWave cellular networks in which D2D relays are used to route around blockages. Using stochastic geometry, we derived coverage probability and spectral efficiency models for the downlink of D2D relay-assisted mmWave cellular networks under the assumption that D2D transmissions are employed on uplink spectrum. Our analytical and simulation results provide numerous important insights on the performance of mmWave cellular links, mmWave D2D links, microwave D2D links, and D2D relay-assisted mmWave cellular networks:

\begin{itemize}
\item Dominant interferer analysis provides a reasonably tight upper bound on the coverage probabilities obtained by PPP-based BS models for mmWave links, and this upper bound is tighter for higher obstacle densities. Dominant interferer analysis becomes inaccurate at extremely high BS densities/low ISDs, particularly when the obstacle density is low.
\item Obstacles play an important role in mitigating interference on mmWave links. On mmWave links, higher obstacle densities lead to higher coverage probabilities at high SINR thresholds. 
\item MmWave cellular links are noise-limited at light to moderate BS densities, but become interference-limited at higher BS densities. The BS density at which they become interference-limited increases with the obstacle density.
\item The BS density that maximizes the coverage probability increases with the obstacle density in both mmWave cellular networks and relay-assisted mmWave cellular networks.  
\item MmWave and microwave D2D relays improve the downlink coverage probability and SE in various scenarios across a wide range of BS densities and SINR thresholds.
\item Deploying D2D relays on mmWave uplink spectrum has a smaller impact on the cellular uplink resources than deploying them on microwave uplink spectrum.
\end{itemize}

Future work should investigate the benefits of D2D relay-assisted communications on the uplink of mmWave cellular networks. This will be  challenging because, assuming D2D links reuse uplink spectrum, the interference on the cellular and D2D links becomes correlated.

\appendices

\section{Proof of Theorem~\ref{theorem:cov_mmwave_cellular}}
\label{app:proof_cov_mmwave_cellular}
To investigate the outage caused by interference, we need to remove BSs that do not have LOS path to the typical receiver by $p(\boldsymbol{x})$-thinning: a point $\boldsymbol{x}$ is removed from a parent PPP with probability $1-p(\boldsymbol{x})$ \cite{chiu:stochastic_book}. For interferers in FIR, moreover, the retained LOS interferers shall be thinned further since only LOS interferers with main lobe towards the typical receiver matter.

By thinning theory, given the distance $d_0$ between the typical UE and its nearest LOS BS, and SINR threshold $\tau$, the average number of dominant LOS BS falling into the NIR is
\vspace{-8pt}
\begin{equation*}
\begin{split}
\Lambda_\text{b}^{\text{(N)}}(\tau, d_0) =& \int_{\text{NIR}} p(\boldsymbol{x}) \lambda_\text{b} (\mathrm{d} \boldsymbol{x}) \\
=& \int_0^{\phi_\text{u}} \int_{d_0}^{D_3} p_{\text{L}}(x)  \lambda_\text{b} x \mathrm{d}x \mathrm{d}\theta + \int_{\phi_\text{u}}^{2\pi} \int_{d_0}^{D_1} p_{\text{L}}(x)  \lambda_\text{b} x \mathrm{d}x \mathrm{d}\theta \\
=& \frac{ \lambda_\text{b} c}{\beta^2} \Big( \phi_\text{u} \big( \beta d_0 e^{-\beta d_0} - \beta D_3 e^{-\beta D_3} + e^{-\beta d_0} - e^{-\beta D_3} \big) \\
& + (2\pi -\phi_\text{u}) \big( \beta d_0 e^{-\beta d_0} - \beta D_1 e^{-\beta D_1} + e^{-\beta d_0} - e^{-\beta D_1} \big) \Big).\\
\end{split}
\vspace{-8pt}
\end{equation*}
Average number of dominant LOS BS falling into the FIR can be obtained similarly. Note that for FIR, only the BSs with boresight towards typical UE will be retained, i.e. a LOS BS in FIR will be retained as a dominant interferer with probability $\frac{\phi_\text{b}}{2\pi}$. We have
\begin{equation*}
\begin{split}
\Lambda_\text{b}^{\text{(F)}}(\tau, d_0) =& \frac{\phi_\text{b}}{2\pi} \int_{\text{FIR}} p(\boldsymbol{x}) \lambda_\text{b} (\mathrm{d} \boldsymbol{x}) \\
=& \frac{\phi_\text{b}}{2\pi} \Big( \int_0^{\phi_\text{u}} \int_{D_3}^{D_4} p_{\text{L}}(x)  \lambda_\text{b} x \mathrm{d}x \mathrm{d}\theta + \int_{\phi_\text{u}}^{2\pi} \int_{D_1}^{D_2} p_{\text{L}}(x)  \lambda_\text{b} x \mathrm{d}x \mathrm{d}\theta \Big)\\
=& \frac{ \phi_\text{b} \lambda_\text{b} c}{2\pi \beta^2} \Big( \phi_\text{u} \big( \beta D_3 e^{-\beta D_3} - \beta D_4 e^{-\beta D_4} + e^{-\beta D_3} - e^{-\beta D_4} \big) \\
& + (2\pi -\phi_\text{u}) \big( \beta D_1 e^{-\beta D_1} - \beta D_2 e^{-\beta D_2} + e^{-\beta D_1} - e^{-\beta D_2} \big) \Big).\\
\end{split}
\end{equation*}

Obviously, the coverage probability by dominant interferer analysis is the null probability, i.e. no retained interferer (BS) falls into both the NIR and FIR. For given outage SINR threshold $\tau$, we have $p_\text{c}(\tau,d_0) = e^{-\Lambda_\text{b}^{\text{(N)}}(\tau, d_0) -\Lambda_\text{b}^{\text{(F)}}(\tau, d_0)}$, 
the coverage probability for a typical UE is
\begin{equation*}
\begin{split}
p_\text{c}(\tau) =& \mathbb{E}_{d_0} \left[ p_\text{c}(\tau,d_0) \right] = \int_{x>0} p_\text{c}(\tau,x) f_{d_0}(x) \mathrm{d}x \\
=& 2\pi\lambda_\text{b} c \int_{x>0} x e^{ -\Lambda_\text{b}^{\text{(N)}}(\tau, x) -\Lambda_\text{b}^{\text{(F)}}(\tau, x) -\Lambda(x,\lambda_\text{b}) -\beta x}\mathrm{d}x.
\end{split}
\end{equation*}

\section{Proof of Theorem~\ref{theorem:cov_microwave_d2d}}
\label{app:proof_cov_microwave_d2d}

Denote $d_{0,\text{N}}$ and $d_{0,\text{L}}$ the distances from the typical UE to its nearest NLOS candidate relay UE and nearest LOS candidate relay UE. The nearest NLOS candidate relay UE will be selected as relay if $A_\text{L}^{-1} d_{0,\text{L}}^{-\alpha_\text{L}} < A_\text{N}^{-1} d_{0,\text{N}}^{-\alpha_\text{N}}$. For $\lambda_\text{r} > 0$, we have
\begin{equation}
\begin{split}
\mathbb{P}(A_\text{L}^{-1} d_{0,\text{L}}^{-\alpha_\text{L}} <& A_\text{N}^{-1} d_{0,\text{N}}^{-\alpha_\text{N}}) = \mathbb{P}(d_{0,\text{N}} < \tilde{a} d_{0,\text{L}}^{\alpha_\text{L}/\alpha_\text{N}}) \\
=& \int_0^\infty f_{d_{0,\text{L}}}(x) \int_0^{\tilde{a} x^{\alpha_\text{L}/\alpha_\text{N}}} f_{d_{0,\text{N}}}(y) \mathrm{d}y \mathrm{d}x \\
=& \int_0^\infty f_{d_{0,\text{L}}}(x) F_{d_{0,\text{N}}}(\tilde{a} x^{ \frac{\alpha_\text{L}}{\alpha_\text{N}} }) \mathrm{d}x \\
=& 1 - \int_0^\infty f_{d_{0\text{r,L}}}(x) e^{-\pi \lambda_\text{r} \tilde{a}^2 x^{2 \frac{\alpha_\text{L}}{\alpha_\text{N}}} +\Lambda( \tilde{a} x^{\frac{\alpha_\text{L}}{\alpha_\text{N}}}, \lambda_\text{r})} \mathrm{d}x \\
=& 1 - 2\pi\lambda_\text{r} c \int_0^\infty x e^{-\Lambda(x,\lambda_\text{r}) -\beta x -\pi \lambda_\text{r} \tilde{a}^2 x^{2 \frac{\alpha_\text{L}}{\alpha_\text{N}}} +\Lambda(\tilde{a} x^{\frac{\alpha_\text{L}}{\alpha_\text{N}}}, \lambda_\text{r})} \mathrm{d}x \triangleq S_\text{N},
\end{split}
\end{equation}
where $\Lambda(d,\lambda)$ is given as \eqref{eq:lambda_big_generic}. $S_\text{N}$ here is defined as the probability that UE associates to a LOS candidate relay UE. Similarly we can obtain $S_\text{L}$, the probability that UE associates to a NLOS candidate relay UE, and $S_\text{L} + S_\text{N} = 1$.

The coverage probability of a D2D link, given that the D2D link is NLOS (or, UE associates to the nearest NLOS candidate relay UE) and has length $d_{0,\text{N}}$, is 
\begin{equation}
\begin{split}
p_\text{c,N}(\tau,d_{0,\text{N}}) =& \mathbb{P}[\frac{h A_\text{N}^{-1} d_{0,\text{N}}^{-\alpha_\text{N}}}{\sigma^2 + I}>\tau ] \\
=& \mathbb{E}_{I} \big[ \mathbb{P}[h > \tau A_\text{N} d_{0,\text{N}}^{\alpha_\text{N}}(\sigma^2 + I) \big|I] \big] \\
\stackrel{(a)}{=}& \mathbb{E}_{I}[e^{-\mu\tau A_\text{N} d_{0,\text{N}}^{\alpha_\text{N}}(\sigma^2+I)} ] = e^{-\mu\tau A_\text{N} d_{0,\text{N}}^{\alpha_\text{N}}\sigma^2} \mathcal{L}_{I} (\mu \tau A_\text{N} d_{0,\text{N}}^{\alpha_\text{N}}),
\end{split}
\end{equation}
where $(a)$ follows the fact that $h \sim \exp(\mu)$, $\mathcal{L}_{I}(\cdot)$ denotes Laplace transform of interference $I$. 
With Rayleigh fading interference, we have \cite{haenggi:stochastic_graph} $\mathcal{L}_{I}(s) = \exp \left( -\pi \rho \lambda_\text{b} s^{2/\alpha_\text{N}} \frac{2\pi/\alpha_\text{N}}{\sin(2\pi/\alpha_\text{N})} \right)$.
Note that the path loss exponent of interference links is assumed to be $\alpha_\text{N}$. Now we have
\begin{equation}
\begin{split}
p_\text{c,N}(\tau) =& \mathbb{E}_{d_{0,\text{N}}} \big[ p_\text{c,N}(\tau,d_{0,\text{N}}) \big] = \int_{x>0} p_\text{c,N}(\tau,x) f_{d_{0,\text{N}}} (x) \mathrm{d}x \\
=& 2\pi\lambda_\text{r} \int_0^\infty x (1 -c e^{-\beta x}) e^{\Lambda(x,\lambda_\text{r}) -\pi\lambda_\text{r} x^2 - \mu\tau\sigma^2 A_\text{N} x^{\alpha_\text{N}}} \mathcal{L}_{I}(\mu\tau A_\text{N} x^{\alpha_\text{N}}) \mathrm{d}x.
\end{split}
\end{equation}

When the nearest LOS candidate relay UE is selected as relay, i.e. $A_\text{L}^{-1} d_{0,\text{L}}^{-\alpha_\text{L}} \geq A_\text{N}^{-1} d_{0,\text{N}}^{-\alpha_\text{N}}$, similarly, we have $p_\text{c,L}(\tau,d_{0,\text{L}}) = e^{-\mu\tau A_\text{L} d_{0,\text{L}}^{\alpha_\text{L}}\sigma^2} \mathcal{L}_{I} (\mu \tau A_\text{L} d_{0,\text{L}}^{\alpha_\text{L}})$, 
and
\begin{equation}
\label{eq:coverage_micro_los_app}
\begin{split}
p_\text{c,L}(\tau) =& \mathbb{E}_{d_{0,\text{L}}} [p_\text{c,L}(\tau,d_{0,\text{L}})] = \int_{x>0} p_\text{c,L}(\tau,x) f_{d_{0,\text{L}}} (x) \mathrm{d}x \\
=& 2\pi \lambda_\text{r} c \int_{x>0} e^{- \Lambda(x, \lambda_\text{r}) -\beta x -\tau \sigma^2 A_\text{L} x^{\alpha_\text{L}}} \mathcal{L}_{I}(\tau A_\text{L} x^{\alpha_\text{L}}) \mathrm{d}x. \\
\end{split}
\end{equation}

By the law of total probability, the microwave D2D link coverage probability for given SINR threshold $\tau$ is $p_\text{c}(\tau) = S_\text{N} p_\text{c,N}(\tau) + S_\text{L} p_\text{c,L}(\tau)$. We thus complete the proof.

\ifCLASSOPTIONcaptionsoff
  \newpage
\fi


\begin{thebibliography}{1}

\bibitem{wu:mmwave}
S. Wu, R. Atat, N. Mastronarde, and L. Liu, ``Coverage analysis of D2D relay-assisted millimeter-wave cellular networks,'' in \textit{IEEE Wireless Communications and Networking Conference (WCNC)}, Mar. 2017. 
  
\bibitem{rappaport:mmwave_will_work}
T. S. Rappaport, S. Sun, R. Mayzus, H. Zhao, Y. Azar, K. Wang, G. N. Wong, J. K. Schulz, M. Samimi, and F. Gutierrez, ``Millimeter wave mobile communications for 5g cellular: it will work!'' \textit{IEEE Access}, vol. 1, pp. 335-349, May 2013.
  
\bibitem{3gpp:mmwave_TR}
3GPP, ``Technical report 38.900: Channel model for frequency spectrum above 6 GHz (Release 14)'' v14.1.0. Sept. 2016. [Online]. Available: http://www.3gpp.org/DynaReport/38900.htm

\bibitem{fcc:mmwave}
Federal Communications Commission (FCC), ``Spectrum frontiers R\&O and FNPRM,'' July 2016. [Online]. Available: https://www.fcc.gov/document/spectrum-frontiers-ro-and-fnprm

\bibitem{zhao:28ghz}
H. Zhao, \textit{et al.}, ``28 GHz millimeter wave cellular communication measurements for reflection and penetration loss in and around buildings in New York City,'' in \textit{Proc. IEEE ICC}, 2013, pp. 5163-5167.

\bibitem{collonge:60ghz}
S. Collonge, G. Zaharia, and G. E. Zein, ``Influence of human activity on wide-band characteristics of the 60GHz indoor radio channel,'' \textit{IEEE Trans. Wireless Commun.}, vol. 3, no. 6, pp. 2369-2406, Nov. 2004.

\bibitem{rappaport:mmwave_propagation}
T. S. Rappaport, \textit{et al.}, ``Broadband millimeter-wave propagation measurements and models using adaptive-beam antennas for outdoor urban cellular communications,'' \textit{IEEE Trans. Antennas and Propag.}, vol. 61, no. 4, pp. 1850-1859, Apr. 2013.

\bibitem{niu:mmwave_survey}
Y. Niu, Y. Li, D. Jin, L. Su, A. V. Vasilakos, ``A survey of millimeter wave (mmWave) communications for 5G: opportunities and challenges,'' ArXiv e-print, Feb. 2015. [Online]. Available: http://arxiv.org/abs/1502.07228

\bibitem{3gpp:relay}
3GPP, ``Technical specification 36.216: Physical layer for relaying operation'', v13.0.0. Jan. 2016.

\bibitem{wi-fi-direct}
Wi-Fi Alliance, ``Wi-fi direct,'' [Online]. Available: http://www.wi-fi.org/discover-wi-fi/wi-fi-direct

\bibitem{lte-direct}
Qualcomm, ``Lte direct,'' [Online]. Available: https://www.qualcomm.com/invention/research/projects/lte-direct

\bibitem{feng:device}
D. Feng, L. Lu, Y. Yuan-Wu, G. Li, S. Li, and G. Feng, ``Device-to-device communications in cellular networks,'' \textit{IEEE Commun. Mag.}, vol. 52, no. 4, pp. 49-55, 2014.

\bibitem{3gpp:prose}
3GPP, RP-142311: ``Work Item Proposal for Enhanced LTE Device to Device Proximity Services,'' Dec. 2014.

\bibitem{wu:flashlinq}
X. Wu, \textit{et al.} ``FlashLinQ: A Synchronous Distributed Scheduler for Peer-to-Peer Ad Hoc Networks,'' \textit{IEEE/ACM Trans. on Networking}, vol. 21, no. 4, pp. 1215-1228, August 2013.

\bibitem{golrezaei:cache}
N. Golrezaei, A. F. Molisch, A. G. Dimakis, and G. Caire, ``Femtocaching and device-to-device collaboration: A new architecture for wireless video distribution,'' \textit{IEEE Commun. Mag.}, vol. 51, no. 4, pp. 142-149, 2013.

\bibitem{corson:prose}
M. S. Corson, R. Laroia, J. Li, V. Park, T. Richardson, and G. Tsirtsis, ``Toward proximity-aware internetworking,'' \textit{IEEE Wireless Commun}, vol. 17, no. 6, pp. 26-33, Dec. 2010.

\bibitem{zhou:intracluster} 
B. Zhou, H. Hu, S.-Q. Huang, and H.-H. Chen, ``Intracluster device-to-device relay algorithm with optimal resource utilization,'' \textit{IEEE Trans. on Vehicular Technology}, vol. 62, no. 5, pp. 2315-2326, June 2013.

\bibitem{pratas:underlay}
N. K. Pratas and P. Popovski, ``Underlay of low-rate machine-type D2D links on downlink cellular links,'' in \textit{2014 IEEE International Conference on Communications (ICC)}, June 2014.

\bibitem{mastronarde:to_relay}
N. Mastronarde, V. Patel, J. Xu, L. Liu, and M. van der Shaar, ``To relay or not to relay: learning device-to-device relaying strategies in cellular networks,'' \textit{IEEE Trans. Mobile Computing}, vol. 15, issue 6, pp. 1569-1585, Jan. 2016.

\bibitem{roh:mmwave_beamforming}
W. Roh, \textit{et al.}, ``Millimeter-wave beamforming as an enabling technology for 5G cellular communications: theoretical feasibility and prototype results,'' \textit{IEEE Commun. Mag.}, vol. 52, issue 2, pp. 106-113. Feb. 2014.

\bibitem{venugopal:d2d_mmwave}
K. Venugopal, M. C. Valenti, and R. W. Heath, Jr., ``Interference in finite-sized highly dense millimeter wave networks,'' in \textit{Proc. IEEE Inf. Theory Appl. Workshop (ITA)}, pp. 175–180. Feb. 2015.

\bibitem{chiu:stochastic_book}
S. Chiu, D. Stoyan, W. Kendall, and J. Mecke, \textit{Stochastic Geometry and Its Applications}, 3rd edition. Wiley, 2013.
 
\bibitem{baccelli:stochastic_wireless_book} 
F. Baccelli and B. Blaszczyszyn, \textit{Stochastic Geometry and Wireless Networks, Volume I-Theory}. Delft, The Netherlands: NOW, 2009.

\bibitem{andrews:tractable}
J. G. Andrews, F. Baccelli, and R. Krishna Ganti, ``A tractable approach to coverage and rate in cellular networks,'' \textit{IEEE Trans. Commun.}, vol. 59, no. 11, pp. 3122-3134, Nov. 2011.

\bibitem{haenggi:stochastic_graph}
M. Haenggi, \textit{et al.}, ``Stochastic geometry and random graphs for the analysis and design of wireless networks,'' \textit{IEEE J. Sel. Areas Commun.}, vol. 27, no. 7, pp. 1029-1046, Aug. 2009.

\bibitem{qiao:mmwave_d2d}
J. Qiao, \textit{et al.}, ``Enabling device-to-device communications in millimeter-wave 5G cellular networks,'' \textit{IEEE Commun. Mag.}, vol. 53, issue 1, pp. 209-215, Jan. 2015.

\bibitem{lin:connectivity}
X. Lin and J. G. Andrews, ``Connectivity of millimeter wave networks with multi-hop relaying,'' \textit{IEEE Wireless Commun. Letters}, vol.4, no. 2, pp. 209-212, Apr. 2015.

\bibitem{wei:mmwave_d2d_relay}
N. Wei, X. Lin, and Z. Zhang, ``Optimal relay probing in millimeter wave cellular systems with device-to-device relaying,'' \textit{IEEE Trans. Vehicular Technology}, vol. PP, issue 99, pp. 1-6, Apr. 2016.

\bibitem{turgut:relay_energy}
E. Turgut, and M. C. Gursoy, ``Energy efficiency in relay-assisted mmwave cellular networks,'' in \textit{Proc. IEEE 84rd Veh. Technol. Conf. (VTC Fall)}, pp. 1-5, Sept. 2016.

\bibitem{xie:relay_assisted}
B. Xie, Z. Zhang, and R. Q. Hu, ``Performance study on relay-assisted millimeter wave cellular networks,'' in \textit{Proc. IEEE 83rd Veh. Technol. Conf. (VTC Spring)}, pp. 1-5, May 2016.

\bibitem{bai:coverage}
T. Bai, R. Vaze, and R. Heath, ``Analysis of blockage effects on urban cellular networks,'' \textit{IEEE Trans. on Wireless Commun.}, vol. 13, no. 9, pp. 5070-5083, Sept. 2014.

\bibitem{lee:bs_position}
C. H. Lee, C. Y. Shih, and Y. S. Chen, ``Stochastic geometry based models for modeling cellular networks in urban areas,'' \textit{Wireless Networks}, vol. 19, issue 6, pp. 1063-1072, Aug. 2013.

\bibitem{3gpp:36300}
3GPP, ``Technical Specification 36.300: Evolved Universal Terrestrial Radio Access Network (E-UTRAN); Overall description (Release 13),'' v13.4.0. June 2016. [Online]. Available: http://www.3gpp.org/DynaReport/36300.htm

\bibitem{lin:d2d_overview}
X. Lin, J. G. Andrews, A. Ghosh, and R. Ratasuk, ``An overview of 3GPP device-to-device proximity services,'' \textit{IEEE Commun. Mag.}, vol. 52, issue 4, pp. 40-48, Apr. 2014.

\bibitem{weber:overview}
S. Weber, J. G. Andrews, and N. Jindal, ``An overview of the transmission capacity of wireless networks,'' \textit{IEEE Trans. Commun.}, vol. 58, no. 12, pp. 3593-3604, Dec. 2010.

\bibitem{andrews:mmwave}
J. G. Andrews, T. Bai, M. N. Kulkarni, A. Alkhateeb, A. K. Gupta, and R. W. Heath, Jr., ``Modeling and analyzing millimeter wave cellular systems,'' \textit{IEEE Trans. Commun.}, vol. 65, no. 1, pp. 403-430, Jan. 2017.

\bibitem{novlan:uplink}
T. Novlan, H. Dhillon, and J. G. Andrews, ``Analytical modeling of uplink cellular networks,'' \textit{IEEE Trans. Wireless Commun.}, vol. 12, no. 6, pp. 2669-2679, Jun. 2013.

\bibitem{Rangan:mmwave_noise_limited}
S. Rangan, T. S. Rappaport, and E. Erkip, ``Millimeter wave cellular wireless networks: potentials and challenges,'' \textit{Proc. IEEE}, vol. 102, no. 3, pp. 366-385, Mar. 2014.

\bibitem{Akdeniz:mmwave_evaluation}
M. R. Akdeniz \textit{et al.}, ``Millimeter wave channel modeling and cellular capacity evaluation,'' \textit{IEEE J. Sel. Areas Commun.}, vol. 32, no. 6, pp. 1164-1179, June 2014.

\bibitem{bai:coverage_2015}
T. Bai, and R. W. Heath Jr., ``Coverage and rate analysis for millimeter-wave cellular networks,'' \textit{IEEE Trans. Wireless Commun.}, vol. 14, no. 2, pp. 1100-1114, Aug. 2014.

\bibitem{3gpp:d2d_843}
3GPP, ``Technical report 36.843: Study on LTE device to device proximity services; radio aspects,'' v12.0.1. Mar. 2014. [Online]. Available: http://www.3gpp.org/dynareport/36843.htm


\bibitem{kulkarni:coverage}
M. N. Kulkarni, S. Singh, and J. G. Andrews, ``Coverage and rate trends in dense urban mmWave cellular networks,'' in \textit{Proc. IEEE Globecom}, 2014, pp. 3809-3814.

\bibitem{chicago_data}
M. N. Kulkarni. \textit{MATLAB codes for converting building location data from shape files to MAT files}. [Online]. Available: https://goo.gl/Ie39k7

\end{thebibliography}
\end{document}